\documentclass[sigconf]{acmart}
\AtBeginDocument{%
  \providecommand\BibTeX{{%
    \normalfont B\kern-0.5em{\scshape i\kern-0.25em b}\kern-0.8em\TeX}}}

\usepackage{color}
\usepackage{url}
\usepackage{mathtools}

    \usepackage{graphicx}

\usepackage{float}

\usepackage{xcolor}
\usepackage{graphics,float}
\usepackage{graphicx}
\usepackage{caption}
\usepackage{mathtools}
\usepackage{setspace}
\usepackage{enumitem}

\DeclarePairedDelimiter\floor{\lfloor}{\rfloor}

\usepackage{multirow}
\usepackage{color, colortbl}
\definecolor{Gray}{gray}{0.9}
\definecolor{Red}{rgb}{1,0.7,0.7 }
\definecolor{Blue}{rgb}{0.6,0.9,0.97 }
\definecolor{Green}{rgb}{0.55,0.92,0.55 }

\makeatletter
\newcommand{\vast}{\bBigg@{3}}
\newcommand{\Vast}{\bBigg@{5}}
\makeatother
\usepackage{amsmath}
\usepackage[ruled,linesnumbered]{algorithm2e}

\def\R{\mathcal{R}}    

\newtheorem{prop}{Proposition}

\newtheorem{remark}{Remark}

\usepackage{color,soul}

\copyrightyear{2021}
\acmYear{2021}
\setcopyright{acmcopyright}\acmConference[RAID '21]{24th International Symposium on Research in Attacks, Intrusions and Defenses}{October 6--8, 2021}{San Sebastian, Spain}
\acmBooktitle{24th International Symposium on Research in Attacks, Intrusions and Defenses (RAID '21), October 6--8, 2021, San Sebastian, Spain}
\acmPrice{15.00}
\acmDOI{10.1145/3471621.3471853}
\acmISBN{978-1-4503-9058-3/21/10}

\begin{document}


\title{The Curse of Correlations for Robust Fingerprinting of Relational Databases}


\author{Tianxi Ji}
\email{txj116@case.edu}
\affiliation{
  \institution{Case Western Reserve University}
  \city{Cleveland}
  \state{Ohio}
  \country{USA}
  }
\author{Emre Yilmaz}
\email{yilmaze@uhd.edu}
\affiliation{
  \institution{University of Houston-Downtown}
  \city{Houston}
  \state{TX}
  \country{USA}}
\author{Erman Ayday}
\email{exa208@case.edu}
\affiliation{  \institution{Case Western Reserve University}
  \city{Cleveland}
  \state{Ohio}
  \country{USA}}
\author{Pan Li}
\email{lipan@case.edu}
\affiliation{  \institution{Case Western Reserve University}
  \city{Cleveland}
  \state{Ohio}
  \country{USA}}

\begin{abstract}
Database fingerprinting   have been widely adopted to prevent unauthorized sharing of data  and identify the source of data leakages. 
Although existing schemes are robust against common attacks, like random bit flipping and subset attack, their robustness degrades significantly if   attackers utilize the inherent correlations among database entries. 
In this paper, we first demonstrate the vulnerability of   existing database fingerprinting schemes by identifying different correlation attacks: column-wise correlation attack, row-wise correlation attack, and the integration of them. 
To provide robust fingerprinting against the identified correlation attacks, we then develop mitigation techniques, which can work as post-processing steps for any off-the-shelf database fingerprinting schemes. 
The proposed mitigation techniques also preserve the utility of the fingerprinted database considering different utility metrics. 
We empirically investigate the impact of the identified correlation attacks and the performance of mitigation techniques using   real-world relational databases. 
Our results show (i) high success rates of the identified correlation  attacks  against   existing fingerprinting schemes (e.g., the integrated correlation attack can distort 64.8\%  fingerprint bits by just modifying 14.2\% entries in a fingerprinted database), and (ii) high robustness of the proposed mitigation techniques (e.g., with the mitigation techniques, the   integrated correlation attack can only distort $3\%$ fingerprint bits). 

\end{abstract}





  \maketitle




\section{Introduction}

Relational databases (or relations) have become the most popular database systems ever since 1970s. A relation is defined as a set of data records  with the same attributes~\cite{codd2002relational}. 
Constructing and sharing of the relations are critical to the vision of a data-driven future that benefits all human-beings. {\color{black}It   supports broader range of tasks in real-life than just sharing database statistics or  machine learning models trained from the database.} 
For example, a relational database owner (who collects data from individuals and constructs the dataset) can benefit from outsourced computation (e.g., from service providers (SP) like Amazon Elastic Compute Cloud), let other SPs analyze its data  (e.g., for personal advertisements), or  exchange datasets for collaborative research after data use agreements.


Most of the time, sharing a database with an authorized  SP (who is authorized to receive/use the database) is done via consent of the database owner. However, when such   databases are shared or leaked beyond the authorized SPs, individuals' (people who contribute their data in the  database) privacy is violated, and hence preventing unauthorized sharing of databases is of great importance. Thus, database owners want to (i) make sure that shared data is used only by the authorized parties for specified purposes and (ii) discourage such parties from releasing the received datasets to other unauthorized third parties (either intentionally or unintentionally). Such data breaches cause financial and reputational damage to database owners. 
For instance, it is reported that the writing site Wattpad suffered a major data breach in July 2020; over 270 million individuals'  data were sold on a third party   forum in the darknet \cite{Wattpadd66:online}. 
Therefore, identifying the source of data breaches is crucial for database owners to hold the identified party responsible.

Digital fingerprinting is a technology that allows to identify the source of data breaches by embedding a unique mark into each shared copy of a digital object. Unlike digital watermarking, in fingerprinting, the embedded mark must be unique to detect the guilty party who is responsible for the leakage. Although the most prominent usage of fingerprinting is in the multimedia domain \cite{cox1997secure,cox2002digital,johnson2001information}, fingerprinting techniques for databases have also been developed \cite{li2005fingerprinting,guo2006fingerprinting,liu2004block,lafaye2008watermill}. These techniques change database entries at different positions when sharing a database copy with a SP. 
However, existing fingerprinting schemes for databases have been developed to embed fingerprints in continuous-valued numerical entries (floating points) in relations. On the other hand, fingerprinting discrete (or categorical) values is more challenging, since the number of possible values (or instances) for a data point is much fewer. Hence, in such databases, a small change in the value of a data point (as a fingerprint) can significantly affect the utility. 
In addition, existing fingerprinting schemes for databases do not consider various inherent correlations between the data records in a database. A malicious party having a fingerprinted copy of a database can detect and distort the embedded fingerprints using its knowledge about the correlations in the data. For example, the zip codes are strongly correlated with street names in a demographic database  making common fingerprinting schemes venerable to attacks utilizing such correlations. Thus,   to provide robustness against correlation attacks (which utilizes the correlations between attributes and data records to infer the potentially fingerprinted entries), we need to consider  such correlations when developing  fingerprinting schemes for relational database. 

In this work, we first identify correlation attacks against the existing database fingerprinting schemes. Namely, we present column-wise correlation attack, row-wise correlation attack, and the integration of both. To launch these attacks, a malicious SP utilizes its prior knowledge about correlations  between the columns (attributes) of database, statistical relationships between the rows (data records), and the combination of both.  
After launching these attacks on a fingerprinted database, the malicious SP can easily distort the added fingerprint to mislead the fingerprint extraction algorithm and cause the database owner to accuse innocent parties. For example, we show that by changing $14.2\%$   entries in a database, the integration of row- and column-wise correlation attack can distort $64.8\%$ fingerprint bits and cause the database owner falsely accuse innocent SPs with high probability.  
This suggests that existing database fingerprinting schemes are vulnerable to identified correlation attacks, and mitigation techniques are in dire need.

To reduce the identified vulnerability in existing database fingerprinting schemes, we propose   novel mitigation techniques  to provide robust fingerprinting   that can alleviate the   correlation attacks. 
Although,   we describe the proposed techniques for a specific vanilla database fingerprinting scheme \cite{li2005fingerprinting}, they can be applied  to other   schemes as well. In other words, the proposed mitigation techniques can work as  post-processing steps for any off-the-shelf database  fingerprinting schemes   and make them robust against potential attacks that utilize the inherent  data correlations. 
The proposed mitigation techniques utilize database owner's prior knowledge on the column- and row-wise correlations. In particular, to mitigate the column-wise correlation attack, the database owner modifies some of the non-fingerprinted data entries to make the post-processed fingerprinted database have column-wise correlations close to that of her prior knowledge. The data entry modification plans are determined from the solutions to a set of ``optimal transportation'' problems~\cite{courty2016optimal}, {\color{black}each of which transports the mass of the marginal distribution of a specific attribute (column) to make it resemble the reference marginal distribution computed from database owner's prior knowledge while minimizing the transportation cost.} 
To alleviate the row-wise correlation attack, the database owner modifies limited number of non-fingerprinted data entries by solving a combinatorial search problem to make the post-processed fingerprinted database have row-wise statistical relationships that are far away from that of her prior knowledge. We show that even if the malicious SP has access to  the exactly same prior knowledge (i.e., data correlation models) with the database owner,   the proposed mitigation techniques can effectively reduce the vulnerability caused by correlation attacks. 
The proposed mitigation techniques also maintain the utility of the post-processed fingerprinted database   by (i) encoding the database entries as integers, such that the least significant bit (LSB) carries the least information, and adding the fingerprint by only changing the LSBs; and (ii) changing only a small number of database entries.



We use an real-world Census relational database to validate the effectiveness of the proposed robust fingerprinting scheme against  the identified correlation attacks.   In particular,  we show that the  malicious SP can only compromise $3\%$ fingerprint bits, even if it launches the powerful integrated correlation attack on the Census database. Thus,  it will   be held as responsible for data leakage.

We summarize the main contributions of this paper as follows:
\begin{itemize}
  \item We identify correlation attacks 
  that can distort large portion of the fingerprint bits in the  existing database fingerprinting scheme and cause the database owner to accuse innocent SPs with high probability.
  \item We propose robust fingerprinting scheme that involves novel mitigation techniques to alleviate the impact of the identified correlation attacks. The proposed mitigation techniques can work as post-processing steps for any off-the-shelf database fingerprinting schemes.
  \item We   investigate the impact of the identified correlation attacks and the proposed mitigation techniques on an real-world relational database. We show that the correlation attacks are more powerful than traditional attacks, because they can  distort more   fingerprint bits with less utility loss. On the other hand,   the  mitigation techniques can effectively alleviate these attacks and maintain database utility even if  the malicious SP   uses data correlation models that are directly calculated from the data.
  
\end{itemize}

The rest of this paper is organized as follows.  We review related works on existing fingerprinting schemes in Section \ref{sec:related_work}, which is followed by the description on the considered vanilla fingerprinting scheme in Section \ref{sec:vanilla}. In Section \ref{sec:system-threat}, we present the system and threat models, and evaluation metrics. Section \ref{sec:corr_attacks} introduces the identified correlation attacks. In Section \ref{sec:robust_fp}, we develop   robust fingerprinting against the identified attacks. We evaluate the impact of correlation attacks and the performance of the proposed mitigation techniques in Section \ref{sec:eva}. Finally, Section \ref{sec:conclusion} concludes the paper.

\section{Related Work}
\label{sec:related_work}


We first briefly review the works on multimedia fingerprinting, and then focus on existing works on fingerprinting relational database.

Large volume of research on  watermarking and fingerprinting  have targeted multimedia, e.g.,  images~\cite{van1994digital,gonge2013robust}, audio \cite{bassia2001robust,  kirovski2001robust},  videos~\cite{swanson1998multiresolution}, and text documents~\cite{Brassil94hidinginformation,Brassil1995}. Such works benefit from the high redundancy in
multimedia, such that the inserted watermark or fingerprint is imperceptible for human beings.    However, the aforementioned multimedia fingerprinting techniques cannot be applied to fingerprint relational databases. The reason is that a database fingerprinting scheme should be robust against common database operations, such as union, intersection, and updating, whereas multimedia fingerprinting schemes are designed to be robust against operations, like compression and formatting.

Database fingerprinting schemes are usually discussed together with database watermarking schemes \cite{li2003constructing} due to their similarity. 
In the seminal work \cite{agrawal2003watermarking}, Agrawal et al. introduce a watermarking framework for relations with numeric attributes by assuming that the database consumer can tolerate a small amount of error in the watermarked databases. Then, based on \cite{agrawal2003watermarking}, some database fingerprinting schemes have been devised. Specifically, Guo et al. \cite{guo2006fingerprinting} develop a two-stage fingerprinting scheme: the first stage is used to prove database ownership, and the second stage is designed for extracted fingerprint verification. Li et al.~\cite{li2005fingerprinting} develop a database fingerprinting scheme by extending \cite{agrawal2003watermarking} to enable the insertion and extraction of arbitrary bit-strings in relations. Furthermore, the authors provide an extensive robustness analysis (e.g., about the upper bound on the probability of detecting incorrect but valid fingerprint from the pirated database) of their scheme. Although \cite{guo2006fingerprinting,li2005fingerprinting} pseudorandomly determine the fingerprint positions in a database, they are not robust against our identified correlation attacks. 
In this paper, we consider \cite{li2005fingerprinting}  as the vanilla fingerprinting scheme and corroborate its vulnerability against correlation attacks. 
Additionally, Liu et al. \cite{liu2004block} propose a  database fingerprinting scheme by dividing  the relational database into blocks and ensuring that certain bit positions of the data at certain blocks contain specific values. In~\cite{liu2004block},  since the fingerprint is embedded block-wise, it is more susceptible to attacks utilizing correlations in the data. 
As a result, incorporating data correlations in database fingerprint schemes is critical to provide robustness against correlation attacks.

Recently,  Yilmaz et al. \cite{yilmaz2020collusion} develop a probabilistic fingerprinting scheme by explicitly considering the correlations (in terms of conditional probabilities) between data points in data record of a single individual. Ayday et al.~\cite{ayday2019robust} propose an optimization-based fingerprinting scheme for sharing  personal sequential data by minimizing the probability of collusion attack with data correlation being one of the constraints. Our work differs from these works since we focus on  developing robust fingerprint scheme for relational databases, which (i) contain large amount of data records from different individuals, (ii) include both column- and row-wise correlations, and (iii) have different utility requirements.




\section{The Vanilla Fingerprint Scheme}
\label{sec:vanilla}

In this work, we consider the  fingerprinting scheme proposed in \cite{li2005fingerprinting} as the vanilla scheme, for which we show the vulnerability and develop the proposed scheme.  
Assume a database owner shares her data with multiple service providers (SPs). The fingerprint  of a specific SP is obtained using a cryptographic hash function, whose input is the concatenation of the database owner's secret key and the SP's public series number. For fingerprint insertion, the vanilla scheme  pseudorandomly selects one bit position of one attribute of some data records in the database and replaces those bits with the results obtained from the exclusive or (XOR) between mask bits and fingerprint bits, both of which are also determined pseudorandomly.  
For fingerprint extraction, the scheme locates the exact positions of the potentially changed bits,   calculates the fingerprint bits by XORing those bits with the exact mask bits, and finally recovers each bit in the  fingerprint bit-string via majority voting, since each fingerprint bit can be  used to mark many different positions.
To preserve the utility of the fingerprinted database, we will let the vanilla scheme only change the least significant bit (LSB) of selected database entries. {\color{black}For completeness,   we show the steps to insert fingerprint into a database, and the steps to extract fingerprint from a pirated   database, in Algorithms \ref{algo:vanilla-insert} and \ref{algo:vanilla-extract}, respectively.} In Appendix \ref{sec:utility_LSB},  we will empirically  validate  that only changing the LSB indeed leads to higher utility than altering one of the least $k$ significant bits (L$k$SB) of selected entries.

\begin{algorithm}
\small
\SetKwInOut{Input}{Input}
\SetKwInOut{Output}{Output}
\Input{The original relational database $\mathbf{R}$, fingerprinting ratio $\gamma$, database owner's secret key $\mathcal{K}$, pseudorandom number sequence generator $\mathcal{U}$, and the SP's series number $n$ (which can be public).}
\Output{The vanilla fingerprinted relational database $\widetilde{\mathbf{R}}\Big(\mathrm{FP},\emptyset,\emptyset \Big)$.}

Generate the fingerprint bit string of SP $n$, i.e., $f_{\mathrm{SP}_n} = Hash(\mathcal{K}|n)$;

\ForAll{data record $\boldsymbol{r}_i\in\mathbf{R}$}{
\If{$\mathcal{U}_1(\mathcal{K}|\boldsymbol{r}_i.\mathrm{primary\ key})\ \mathrm{mod}\ \gamma = 0$}{
\CommentSty{//fingerprint this data record}

$\mathrm{attribute\_index}\ p = \mathcal{U}_2(\mathcal{K}|\boldsymbol{r}_i.\mathrm{primary\ key})\ \mathrm{mod}\ |\mathcal{F}|$.
\CommentSty{//fingerprint this attribute ($|\mathcal{F}|$ is the cardinality of the attributes set)}

Set $\mathrm{mask\_bit}\ x = 0$, if $\mathcal{U}_3(\mathcal{K}|\boldsymbol{r}_i.\mathrm{primary\ key})$ is even; otherwise set $x=1$.

$\mathrm{fingerprint\_index}\ l=\mathcal{U}_4(\mathcal{K}|\boldsymbol{r}_i.\mathrm{primary\ key})\ \mathrm{mod}\ L$. \CommentSty{//$L$ is the length of the fingerprint bit-string}

$\mathrm{fingerprint\_bit}\ f = f_{\mathrm{SP}_n}(l)$.


$\mathrm{mark\_bit}\ m = x\oplus f$.

 Set the LSB of $\boldsymbol{r}_i.p$ to $m$. 

}

}

 Return $\widetilde{\mathbf{R}}\Big(\mathrm{FP},\emptyset,\emptyset \Big)$.

		\caption{Fingerprint insertion phase of the vanilla fingerprinting scheme \cite{li2005fingerprinting}}
\label{algo:vanilla-insert}
\end{algorithm}

\begin{algorithm}
\small
\SetKwInOut{Input}{Input}
\SetKwInOut{Output}{Output}
\Input{The leaked relational database $\overline{\mathbf{R}}$, fingerprinting ratio $\gamma$, database owner's secret key $\mathcal{K}$, pseudorandom number sequence generator $\mathcal{U}$, and a fingerprint template $(?,?,\cdots,?)$, where $?$ represents unknown value.}
\Output{The extracted fingerprint from the leaked database.}

\ForAll{$l\in[1,L]$}{

$count[l][0] = count[l][1] = 0$. \CommentSty{//$count[l][0]$ and  $count[l][1]$ are number of votes for $f(l)$ to be 0 or 1, respectively.}

}


\CommentSty{//scan all data records and obtain the counts for each fingerprint bit}

\ForAll{data record $\boldsymbol{r}_i\in\mathbf{R}$}{
\If{$\mathcal{U}_1(\mathcal{K}|\boldsymbol{r}_i.\mathrm{primary\ key})\ \mathrm{mod}\ \gamma = 0$}{ 

$\mathrm{attribute\_index}\ p = \mathcal{U}_2(\mathcal{K}|\boldsymbol{r}_i.\mathrm{primary\ key})\ \mathrm{mod}\ |\mathcal{F}|$.

Set $\mathrm{mark\_bit}\ m$ as the LSB of $\boldsymbol{r}_i.p$.

Set $\mathrm{mask\_bit}\ x = 0$, if $\mathcal{U}_3(\mathcal{K}|\boldsymbol{r}_i.\mathrm{primary\ key})$ is even; otherwise set $x=1$.

$\mathrm{fingerprint\_bit}\ f = m\oplus x$.

$\mathrm{fingerprint\_index}\ l=\mathcal{U}_4(\mathcal{K}|\boldsymbol{r}_i.\mathrm{primary\ key})\ \mathrm{mod}\ L$.

$count[l][f] = count[l][f]+1$.



}

}

\CommentSty{//recover the fingerprint bit string}

\ForAll{$l\in[1,L]$}{
\If{count[l][0] = count[l][1]}{return none suspected}

$f(l)=0$ if $count[l][0] > count[l][1] = 0$.

$f(l)=1$ if $count[l][0] < count[l][1] = 0$.

}
 
 Return the extracted fingerprint bit string $f$.

		\caption{Fingerprint extraction phase of the vanilla fingerprinting scheme \cite{li2005fingerprinting}}
\label{algo:vanilla-extract}
\end{algorithm}

In practice, one can choose any database fingerprinting scheme as the vanilla scheme, because our proposed mitigation techniques are independent of the adopted vanilla scheme, and they can be used as post-processing steps on top of any existing database fingerprinting schemes. 
The reason we choose the aforementioned  vanilla scheme  is because
(i) it is shown to have high robustness, e.g., the probability of detecting no fingerprint as a result of random bit flipping attack (a common attack against fingerprinting schemes, as will be discussed in Section~\ref{sec:threat_models}) is upper bounded by ${(|SP|-1)}/{2^L}$, where $|SP|$ is the number of SPs who have received the fingerprinted copies and $L$ is the length of the fingerprint bit-string, (ii) it is shown to be robust even if some fingerprinted entries are identified by a malicious SP, because it applies majority voting on all the fingerprinted entries to extract the fingerprint bit-string, 
and (iii) it can easily be extended to incorporate Boneh-Shaw code \cite{boneh1998collusion} to defend against collusion attacks.  {\color{black}Our developed  robust fingerprinting scheme inherits all the properties of the vanilla  scheme because (i) it uses the  vanilla scheme as the building block and (ii) it does not alter the entries that have already been changed by the vanilla scheme (due to fingerprinting insertion).}



\section{System and Threat Models}
\label{sec:system-threat}







First, we introduce the  nomenclature for different databases obtained by applying various techniques. We denote the database owner's  (i.e., Alice) original database as $\mathbf{R}$, a fingerprinted database shared by her as $\widetilde{\mathbf{R}}$, and the pirated database leaked by a malicious SP as $\overline{\mathbf{R}}$,  respectively. 
Both $\widetilde{\mathbf{R}}$ and $\overline{\mathbf{R}}$ are represented using 3 input parameters showing the techniques that are adopted to generate them.
3 input parameters for $\widetilde{\mathbf{R}}(\alpha,\beta,\eta)$ represent which processes have been applied to the database during fingerprinting, where (i) $\alpha$ represents the vanilla fingerprinting, (ii) $\beta$ represents the proposed mitigation technique against the row-wise correlation attack, and (iii) $\eta$ represents the proposed mitigation technique against the column-wise correlation attack. On the other hand, 3 input parameters for $\overline{\mathbf{R}}(\alpha,\beta,\eta)$ represent which attacks have been conducted by the malicious SP on the fingerprinted database, where (i) $\alpha$ represents the random bit flipping attack, (ii) $\beta$ represents the row-wise correlation attack, and (iii) $\eta$ represents the column-wise correlation attack. We provide the details of these attacks and mitigation techniques in Sections \ref{sec:corr_attacks} and \ref{sec:robust_fp}, respectively. We will also use $\widetilde{\mathbf{R}}$ (or $\overline{\mathbf{R}}$) when referring  to a generic fingerprinted (or pirated) database when its input parameters are clear from the context.

We summarize the frequently used notations in  Table~\ref{main_notations}. For instance, $\widetilde{\mathbf{R}}\big(\mathrm{FP},\mathrm{Dfs_{row}(\mathcal{S}')},\mathrm{Dfs_{col}(\mathcal{J}')}\big)$ represents a fingerprinted database that is generated by applying the vanilla fingerprinting scheme ($\mathrm{FP}$) on the original database $\mathbf{R}$ followed by two proposed mitigation techniques $\mathrm{Dfs_{row}(\mathcal{S}')}$ and $\mathrm{Dfs_{col}(\mathcal{J}')}$ to alleviate the potential correlation attacks (as will be discussed in Sections~\ref{sec:col_defense} and~\ref{sec:row-defense}).  
Here, $\mathcal{S}'$ (or $\mathcal{J}'$) is the database owner's prior knowledge on the row-wise (or column-wise) correlations in the database. 
Similarly, $\overline{\mathbf{R}}\big(\emptyset,\mathrm{Atk_{row}(\mathcal{S})},\mathrm{Atk_{col}}(\mathcal{J})\big)$ represents a pirated database that is generated by a malicious SP by first launching the row-wise correlation attack $\mathrm{Atk_{row}}(\mathcal{S})$, and then the column-wise correlation attack $\mathrm{Atk_{col}}(\mathcal{J})$ (as will be discussed in Section~\ref{sec:col-wise-attack} and \ref{sec:row-wise-atack}, and $\emptyset$ means random bit flipping attack is not applied).  
Here, $\mathcal{S}$ (or $\mathcal{J}$) is the malicious SP's prior knowledge on the row-wise (or column-wise) correlations of the database.   
In general, $\mathcal{S}'\neq\mathcal{S}$ and $\mathcal{J}'\neq\mathcal{J}$, which is referred to as the prior knowledge asymmetry between the database owner and the malicious SP. To the advantage of the malicious SP, we assume that the malicious SP can have access to the correlation models that are directly calculated from the database, i.e., its prior knowledge is as accurate as that of the database owner. 
In the future work, we will also   investigate the scenario where the database owner even has less accurate prior knowledge compared with the malicious SPs.

\begin{table*}[htb]
\begin{center}
 \begin{tabular}{c | c  } 
 \hline
 \hline
  $\mathbf{R}$ & the  original database owned by the database owner (Alice)\\
   \hline
   $\widetilde{\mathbf{R}}$ & a generic fingerprinted database shared by the database owner\\
   \hline
   $\overline{\mathbf{R}}$ & a generic pirated database generated by the malicious SP\\
   \hline
  \multirow{3}{*}{$\widetilde{\mathbf{R}}(\alpha,\beta,\eta)$} & the fingerprinted database obtained by applying   (i) $\alpha$, the vanilla fingerprinting scheme,  \\
  & (ii) $\beta$, the mitigation technique against the row-wise correlation attack,  \\
   & and (iii) $\eta$, the mitigation technique against the column-wise correlation attack  in sequence\\
   \hline
   \multirow{2}{*}{$\overline{\mathbf{R}}(\alpha,\beta,\eta)$} & the pirated database generated by the malicious SP by applying (i) the random bit flipping attack $\alpha$, \\
   & (ii) the row-wise correlation attack $\beta$,  and (iii) the column-wise correlation attack $\eta$ in sequence\\
   \hline
  $\mathcal{S}'$ and $\mathcal{J}'$ & database owner's prior knowledge on the row-wise correlations and column-wise correlations\\
  \hline
    $\mathcal{S}$ and $\mathcal{J}$ & the malicious SP's prior knowledge on the row-wise correlations and column-wise correlations\\
  \hline
    $\widetilde{\mathcal{S}}$ and $\widetilde{\mathcal{J}}$ & the empirical row-wise and column-wise correlations obtained from a generic  fingerprinted database $\widetilde{\mathbf{R}}$\\
    \hline
 $\mathrm{Atk_{col}}(\mathcal{J})$ & the column-wise correlation attack launched by the malicious SP by using prior knowledge $\mathcal{J}$ \\ 
   \hline
  $\mathrm{Dfs_{row}}(\mathcal{S}')$ & the mitigation technique using prior knowledge $\mathcal{S}'$  to alleviate row-wise correlation attack \\
    \hline
  $\mathrm{Dfs_{col}}(\mathcal{J}')$ & the mitigation technique using prior knowledge $\mathcal{J}'$  to alleviate column-wise correlation attack \\
 \hline
 \hline
\end{tabular}
\caption{Frequently used notations in the paper.}
\label{main_notations}
\end{center}
\end{table*}

\subsection{System Model}
\label{system_model}



We present the vanilla   fingerprint system model  in Figure \ref{fig:fp_general_system}. Specifically, we consider the database owner (Alice) with a  {categorical} relational database   $\mathbf{R}$, which includes the data records of $M$ individuals.  We denote the set of attributes of the individuals as $\mathcal{F}$ and the $i$th row (data record) in $\mathbf{R}$ as $\boldsymbol{r}_i$. 
Alice shares her data with multiple service providers (SPs) 
to  receive specific services from them. To prevent unauthorized redistribution of her database by a malicious SP, Alice includes a unique fingerprint in each copy of her database when sharing it with a SP. The fingerprint bit-string associated to SP $i$ ($\mathrm{SP}_i$) is denoted as $f_{\mathrm{SP}_i}$, and the vanilla fingerprinted dataset received by $\mathrm{SP}_i$ is represented as $\widetilde{\mathbf{R}}_{\mathrm{SP}_i}(\mathrm{FP},\emptyset,\emptyset)$. 
Both $f_{\mathrm{SP}_i}$ and $\widetilde{\mathbf{R}}_{\mathrm{SP}_i}(\mathrm{FP},\emptyset,\emptyset)$ are obtained  using the vanilla fingerprint scheme discussed in Section \ref{sec:vanilla}, which changes entries of $\mathbf{R}$ at different positions (indicated by the yellow dots in Figure \ref{fig:fp_general_system}.   
If a malicious SP (e.g., $\mathrm{SP}_i$)   pirates and  redistributes Alice's database, she is able to identify $\mathrm{SP}_i$ as the traitor by extracting its fingerprint in $\widetilde{\mathbf{R}}_{\mathrm{SP}_i}(\mathrm{FP},\emptyset,\emptyset)$ {\color{black}as long as the data entries are not significantly modified (e.g., when less than 80\% entries are changed or removed)}. 


\begin{figure}[htb]
  \begin{center}
     \includegraphics[width= 0.7\columnwidth]{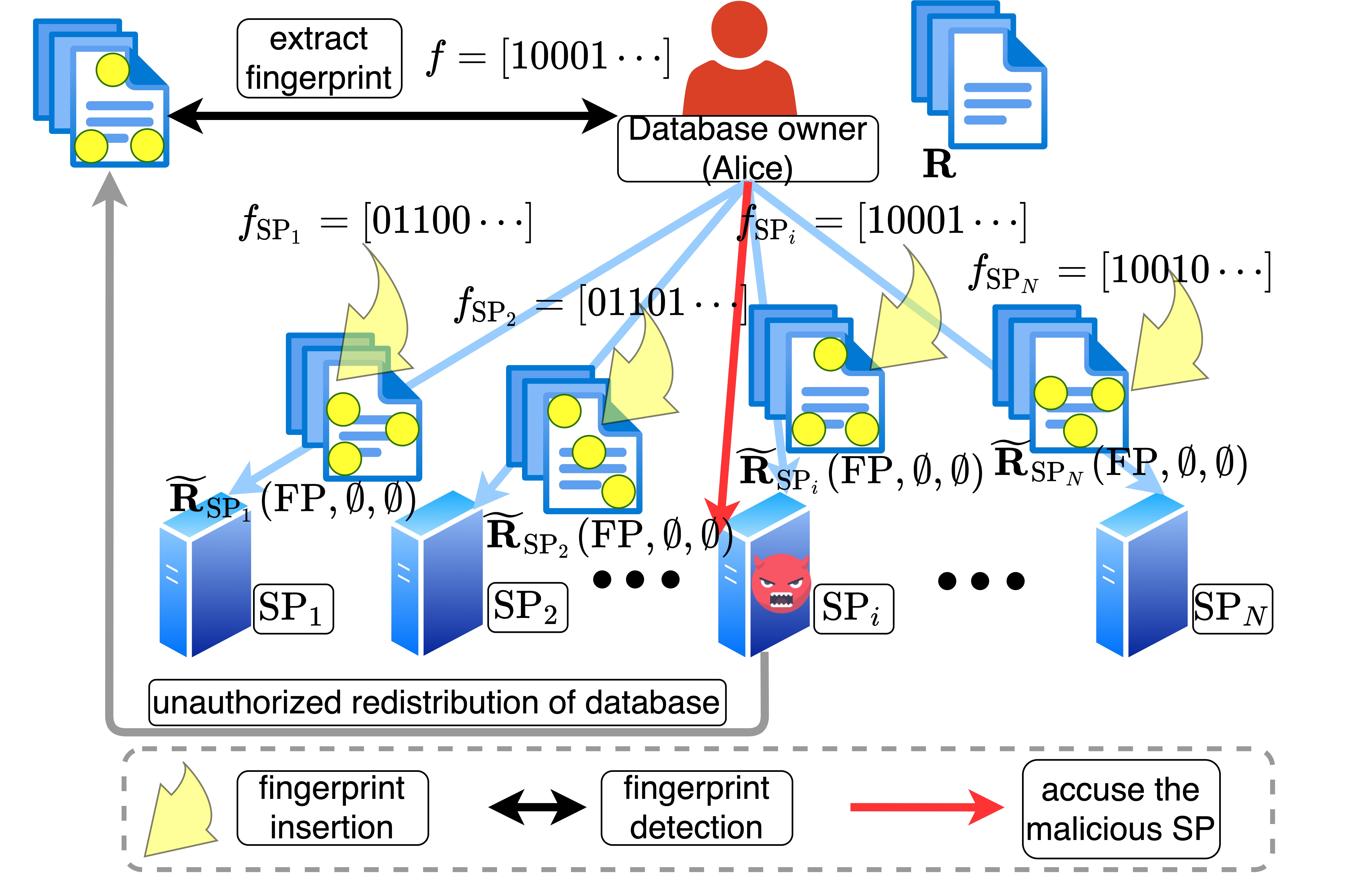}
      \end{center}
  \caption{\label{fig:fp_general_system}  The vanilla fingerprinting system, where Alice adds a unique fingerprint in each shared copy of her database. She is able to identify the malicious SP who pirates and redistributes her database {\color{black}as long as the data entries are not significantly modified (e.g., when less than 80\% of the entries are changed or removed)}.}
\end{figure}

\begin{figure*}[htb]
  \begin{center}
     \includegraphics[width= 1.65\columnwidth]{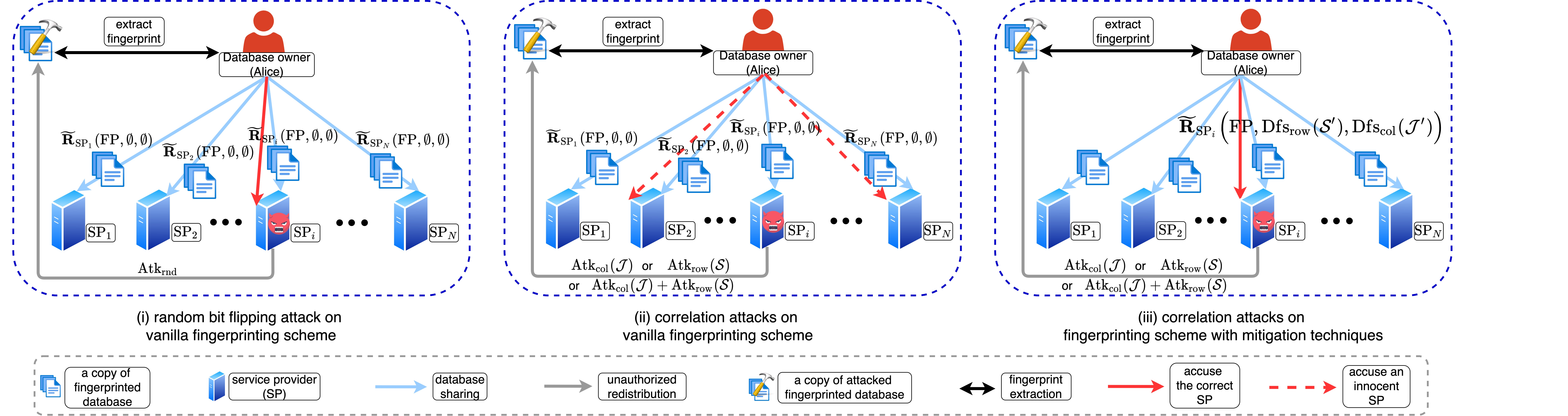}
      \end{center}
  \caption{\label{fig:attacks}   (i) If Alice inserts fingerprinting using the vanilla scheme, and the malicious $\mathrm{SP}_i$ conducts random bit flipping attack, i.e., $\mathrm{Atk_{rnd}}$, on its received copy, i.e.,   $\widetilde{\mathbf{R}}_{\mathrm{SP}_i}(\mathrm{FP},\emptyset,\emptyset)$ and redistributes the data. Then, with high probability, Alice can correctly accuse it for data leakage. (ii) If the malicious $\mathrm{SP}_i$ conducts any correlation attack, e.g., the column-wise correlation attack ($\mathrm{Atk_{col}}(\mathcal{J})$), the row-wise correlation attack ($\mathrm{Atk_{row}}(\mathcal{S})$), or the combination of them, on the vanilla fingerprinted database. Then, with high probability, Alice cannot identify it as the traitor, and she will accuse other innocent SPs. (iii)  If Alice applies the mitigation techniques, i.e., the column-wise correlation defense ($\mathrm{Dfs_{col}}(\mathcal{J}')$) and the row-wise correlation defense ($\mathrm{Dfs_{row}}(\mathcal{S}')$), after the vanilla fingerprinting scheme, and shares $\widetilde{\mathbf{R}}(\mathrm{FP},\mathrm{Dfs_{row}}(\mathcal{S}'),\mathrm{Dfs_{col}}(\mathcal{J}'))$ with $\mathrm{SP}_i$. Then, with high probability, she can correctly identify $\mathrm{SP}_i$ as the traitor even if it conducts any of the correlation attack on its received copy.}
\end{figure*}

\subsection{Threat Model}
\label{sec:threat_models}
Fingerprinted database is subject to various attacks summarized in the following sections. In Figure \ref{fig:attacks}, we show some representative ones that are studied in this paper. 
Note that in all considered attacks, a malicious SP can change/modify most of the entries in $\widetilde{\mathbf{R}}$ to distort the fingerprint (and to avoid being accused). However, such a pirated database will have significantly poor utility (as will be introduced in Section \ref{sec:utility_metrics}).  {\color{black}As discussed in Section \ref{sec:vanilla}, we let the vanilla fingerprint scheme only change the LSBs of data entries to preserve data utility. Thus, all considered attacks also change the LSBs of the selected entries in $\widetilde{\mathbf{R}}$ to distort the fingerprint.}


\subsubsection{Random Bit Flipping Attack} In this attack, to pirate a database, a malicious SP   selects  random entries in $\widetilde{\mathbf{R}}$ and flips their LSBs \cite{agrawal2003watermarking}. The flipped entries are still in the domain of the corresponding attributes. 
The considered vanilla fingerprint scheme is robust against this attack \cite{li2005fingerprinting} as shown in Figure~\ref{fig:attacks}(i),   Alice shares fingerprinted copies of her database by only applying  $\mathrm{FP}$. If a malicious SP ($\mathrm{SP}_i$) tries to distort the fingerprint in $\widetilde{\mathbf{R}}(\mathrm{FP},\emptyset,\emptyset)$ using the random bit flipping attack (i.e., $\mathrm{Atk_{rnd}}$), and then   redistributes it, Alice can still detect $\mathrm{SP}_i$'s fingerprint in the pirated  copy with a high probability, 
and correctly accuse $\mathrm{SP}_i$ of data leakage.

\subsubsection{Subset and Superset Attacks}
In subset attack, a malicious SP generates a pirated copy of $\widetilde{\mathbf{R}}$ by randomly selecting  data records from it.   Superset attack is the dual attack of subset attack. In this attack, the malicious SP mixes $\widetilde{\mathbf{R}}$ with other databases to create a pirated  one. These two attacks are considered to be  weak attacks. For example, for subset attack, to compromise just one specific  bit in the inserted  fingerprint bit-string,  the malicious SP must exclude all records that are marked by that  bit~\cite{li2005fingerprinting}.

\subsubsection{Correlation Attacks}
In correlation attacks, a malicious SP utilizes the inherent correlations in the data to more accurately identify the fingerprinted entries, and hence distort the fingerprint. 
Since the subset and superset attack are not as powerful as the bit flipping attack \cite{yilmaz2020collusion}, we consider developing the correlation attacks based on the random bit flipping attack. In the following, we provide the high level descriptions of two main correlation attacks (details of these attacks are in Section \ref{sec:corr_attacks}).



In column-wise correlation attack, i.e., $\mathrm{Atk_{col}}(\mathcal{J})$, we assume that the malicious SP has prior knowledge about the correlations among each pair of attributes  (or columns in the database) characterized by the set of joint probability distributions $\mathcal{J}$. Once receiving the fingerprinted  database $\widetilde{\mathbf{R}}$, the malicious SP first calculates a new set of joint probability distributions based on  $\widetilde{\mathbf{R}}$.  
Then, it compares 
the new joint distributions with its prior knowledge $\mathcal{J}$, and flips the entries in $\widetilde{\mathbf{R}}$ that leads to large discrepancy between them.





In row-wise correlation attack,  i.e., $\mathrm{Atk_{row}}(\mathcal{S})$, we consider that the individuals belong to  different communities (e.g.,  social circles decided by friendship,  or families determined by kinship), and assume that the malicious SP has the prior knowledge $\mathcal{S}$, which contains (i) each individual's membership  to the communities and (ii) the statistical relationships of   pairs of individuals belonging to  the same community.  
Once it receives the fingerprinted  database $\widetilde{\mathbf{R}}$, the malicious SP first calculates a new set of statistical relationships based on $\widetilde{\mathbf{R}}$, then it compares the newly computed  statistical relationships with   $\mathcal{S}$, and changes the entries that leads to large discrepancy between the two sets of statistical relationships.  

Figure~\ref{fig:attacks}(ii) shows the scenario, where Alice identifies the source of the data leakage wrong and accuses innocent SPs if she uses the vanilla fingerprinting scheme, whereas, $\mathrm{SP}_i$ conducts more advanced attacks to distort the fingerprint. These attacks include   $\mathrm{Atk_{col}}\big(\mathcal{J}\big)$,   $\mathrm{Atk_{row}}\big(\mathcal{S}\big)$, and the combination of them. Finally, Figure~\ref{fig:attacks}(iii) shows that if Alice uses  the proposed mitigation  techniques (i.e., $\mathrm{Dfs_{row}}(\mathcal{S}')$ and  $\mathrm{Dfs_{col}}(\mathcal{J}')$ (as will be discussed in Section~\ref{sec:robust_fp}) after $\mathrm{FP}$ to improve the robustness of the added fingerprint and shares $\widetilde{\mathbf{R}}\big(\mathrm{FP},   \mathrm{Dfs_{row}}(\mathcal{S}'), \mathrm{Dfs_{col}}(\mathcal{J}')\big)$, 
then, even though $\mathrm{SP}_i$ conducts the identified  correlation attacks, Alice can still identify $\mathrm{SP}_i$ to be responsible for leaking the data with high probability.

\subsubsection{Collusion Attack}
Fingerprinted databases are also susceptible to collusion attack, where multiple malicious SPs ally together to generate a pirated database from their unique fingerprinted copies. In   cryptography literature, many works have attempted to develop collusion resistant fingerprinting schemes  \cite{boneh1998collusion,boneh1995collusion,yacobi2001improved,pfitzmann1997asymmetric}.  Our proposed mitigation techniques can also be used with a collusion-resistant vanilla fingerprinting scheme \cite{boneh1998collusion} to   provide some level of robustness against colluding SP.  In this work, we   mainly focus on   correlation attacks from a single-handed malicious SP. We will extend our work in the scenario of colluding SPs in future work. 


\subsection{Fingerprint Robustness Metrics}
\label{sec:robustness_metrics}

The primary goal of a malicious SP is to distort the   fingerprint in $\widetilde{\mathbf{R}}$, thus we consider the following fingerprint  robustness metrics about a pirated  database  $\overline{\mathbf{R}}$ generated by launching attacks on $\widetilde{\mathbf{R}}$.
    
\subsubsection{Number of compromised fingerprint bits}

We formulate the number of compromised fingerprint bits as
    \begin{equation*}
        \mathrm{num_{cmp}} = {\textstyle\sum}_{l = 1}^L \mathbf{1}\{f(l)\neq \overline{f}(l)\},
    \end{equation*}
    where $\mathbf{1}\{\cdot\}$ is the indicator function, $L$ is the length of the fingerprint bit-string, $\overline{f}$ is the extracted  fingerprint bit-string from  $\overline{\mathbf{R}}$, and $f(l)$ (or $\overline{f}(l)$) is the $l$th bit in $f$ (or $\overline{f}$).

\subsubsection{Accusable ranking of a malicious SP}

We quantify the confidence of accusing the correct malicious SP by defining the accusable ranking metric (denoted as $r$) as follows: 
\begin{gather*}
r = 
\begin{cases}
&\mathrm{``uniquely\ accusable"},  \textit{if}\  m_0 >  \sum_{l=1}^L\mathbf{1}\Big\{f_{\mathrm{SP}_i}(l) = \overline{f}(l)\Big\},  \forall \mathrm{SP}_i \in \mathcal{T}  \\
&\mathrm{``top}\ t\ \mathrm{accusable"} , \textit{otherwise}
\end{cases},
\end{gather*}
where   $m_0 = \sum_{l=1}^L\mathbf{1}\{f_{\mathrm{SP_{malicious}}}(l) = \overline{f}(l)\}$ is the number of bit matches between the malicious SP's fingerprint and the extracted fingerprint from the pirated database, and $\mathcal{T}$ is the set of all innocent SPs. Specifically, if the malicious SP has the most bit matches with the extracted fingerprint,   Alice will uniquely accuse it. Otherwise, we compute $t=\frac{\sum_{\mathrm{SP}_i \in \mathcal{T}} \mathbf{1}\Big\{\Big(\sum_{l=1}^L\mathbf{1}\Big\{f_{\mathrm{SP}_i}(l) = \overline{f}(l)\Big\}\Big) \geq m_0\Big\}} {|\mathcal{T}|}\times 100\%$,
which is  the fraction  of innocent SPs having more bit matches with the extracted fingerprint than the malicious SP. For example, if $t = 80\%$, then the malicious SP is only top 80\% accusable, which suggests that Alice will accuse other innocent SPs with high probability. In contrast, if $t = 1\%$, then the malicious SP's accusable ranking increases and makes it among the top 1\% accusable SPs, and Alice will accuse other innocent SPs with low probability. Essentially, a high accusable rank $r$ corresponds to either (i) a ``low $t$" or (ii) the uniquely accusable case. 
{\color{black}As we will show in Section \ref{sec:eva}, for a malicious SP to avoid being accused (i.e., have low accusable rank, or high $t$ value),  it needs to distort more than half of the fingerprint bits. As we will also show via evaluations, the malicious SP can easily achieve this goal if it applies the identified  correlation attacks. Whereas, if it applies the random bit flipping attack, it becomes ``uniquely accusable'' with high probability unless it overdistort the fingerprinted database.}


{\color{black}
According to the vanilla scheme~\cite{li2005fingerprinting}, the probability of extracting a valid fingerprint from a database that does not belong to Alice (i.e., misdiagnosis false hit) is upper bounded by $|\mathrm{SP}|/2^L$, and the probability of extracting an incorrect but valid fingerprint from the fingerprinted database (i.e., misattribution false hit) is upper bounded by $(|\mathrm{SP}|-1)/2^L$. Since these are all negligible probabilities, we do not consider the case in which Alice does not accuse any SP when a copy of her database is leaked in the experiments.
}

\subsection{Utility Metrics}
\label{sec:utility_metrics}


Fingerprinting naturally changes the content of the database, and hence degrades its utility. 
We quantify the utility of a fingerprinted database using the following metrics. 

\subsubsection{Accuracy of $\widetilde{\mathbf{R}}$} 

We quantify the accuracy of $\widetilde{\mathbf{R}}$ as  
\begin{equation*}
      Acc(\widetilde{\mathbf{R}}) = 1-  {\widetilde{\mathbf{R}}\oplus \mathbf{R}}/{M*L},
  \end{equation*}
  where $\oplus$ is the symmetric difference operator that counts   the number of different entries in the fingerprinted and the original databases. $Acc(\widetilde{\mathbf{R}})$ measures the percentage of matched entries between the    fingerprinted and the original databases. 
  
\subsubsection{Preservation of column-wise correlations}

We quantify the preservation of column-wise correlations in the database as
  \begin{gather*}
    \mathrm{P_{col}}(\widetilde{\mathbf{R}}) = 1 - \frac{\sum_{p,q\in \mathcal{F},p\neq q}\sum_{a\in p, b\in q} \mathbf{1}\big\{ \big| \widetilde{J_{p,q}}(a,b)-J_{p,q}(a,b)  \big| \geq\tau_{\mathrm{col}} \big\}}{\sum_{p,q\in \mathcal{F},p\neq q} k_pk_q},
\end{gather*}
  where $p$ and $q$ are two attributes in the attribute set $\mathcal{F}$, $k_p$ (or $k_q$) stands for the number of unique instances of attribute $p$ (or $q$), and $\widetilde{J_{p,q}}(a,b)$ (or $J_{p,q}(a,b)$) is the joint probability that attribute $p$ takes value $a$ and attribute $q$ takes value $b$ in $\widetilde{\mathbf{R}}$ (or $\mathbf{R}$). 
    $\mathrm{P_{col}}$ calculates the fraction of instances of $\big| \widetilde{J_{p,q}}(a,b)-J_{p,q}(a,b)  \big| $ that do not exceed a predetermined threshold $\tau_{\mathrm{col}}$ before and after fingerprinting $\mathbf{R}$. 
  
\subsubsection{Preservation of row-wise correlations}

We quantify the preservation of row-wise correlations in the database as
    \begin{gather*}
        \mathrm{P_{row}}(\widetilde{\mathbf{R}}) = 1 - \frac{\sum_{c=1}^C\sum_{i,j\in\mathrm{comm}_c,i\neq j}\mathbf{1}\big\{\big|\widetilde{s_{i,j}}^{\mathrm{comm}_c}-s_{i,j}^{\mathrm{comm}_c}\big| \geq \tau_{\mathrm{row}}\big\}}{\sum_{c=1}^C n_c(n_c-1)},
    \end{gather*}
where $\mathrm{comm}_c$ represents the set of all individuals in a community $c$,   $\widetilde{s_{i,j}}^{\mathrm{comm}_c}$ (or $s_{i,j}^{\mathrm{comm}_c}$) is the statistical relationship between individual $i$ and $j$ belonging to $\mathrm{comm}_c$ in $\widetilde{\mathbf{R}}$ (or $\mathbf{R}$),   $n_c$ is the number of individuals in $\mathrm{comm}_c$, and $C$ is the number of communities. In essence, $\mathrm{P_{row}}(\widetilde{\mathbf{R}})$ evaluates the fraction of    statistical relationship that has absolute  difference   less than $\tau_{\mathrm{row}}$ in the entire population before and after fingerprinting. 

\subsubsection{Preservation of empirical covariance matrix}

We quantify the preservation of  empirical covariance matrix of the database as
  \begin{equation*}
      \mathrm{P_{cov}} = 1-{||\mathrm{cov}(\widetilde{\mathbf{R}})-\mathrm{cov}(\mathbf{R})||_F}/ {||\mathrm{cov}(\mathbf{R})||_F},
  \end{equation*}
  where $\mathrm{cov}(\mathbf{R}) = \sum_{i=1}^M \boldsymbol{r}_i^T\boldsymbol{r}_i/M$ is the empirical covariance matrix of data records in $\mathbf{R}$. $\mathrm{P_{cov}}$ evaluates the similarity between the covariance matrices of the database before and after fingerprinting. We consider this metric because the fingerprinted database may also be used in data analysis tasks, and  empirical covariance matrix is frequently utilized  to   establish  predictive models, e.g., regression and probability fitting \cite{jolliffe2002springer,browne1992alternative}. Besides,   multivariate data analysis often involves the investigation of  inter-relationships among data records which requires an accurate covariance matrix estimation. 

Note that the utility of the pirated database $\overline{\mathbf{R}}$ generated by the malicious SP can also be quantified using the same metrics, i.e., $Acc(\overline{\mathbf{R}})$, $\mathrm{P_{col}}(\overline{\mathbf{R}})$,    $\mathrm{P_{row}}(\overline{\mathbf{R}})$, and $\mathrm{P_{cov}}(\overline{\mathbf{R}})$. As discussed, a malicious SP can successfully (without being accused) distort the fingerprint easily by over-distorting $\widetilde{\mathbf{R}}$, however, to preserve the data utility, a rational malicious SP will not over-distort a database.


{\color{black}

In addition to the general utility metrics defined above, we will also consider specific statistical utilities, e.g., portion of individuals that have a particular education degree or higher, and the standard deviation of individuals' age distribution. It is noteworthy that  if the general utility metrics are high, it  implicitly suggests high utility for the specific statistical (or other application  related) utilities.}

\section{Identified Correlation Attacks}
\label{sec:corr_attacks}



In the correlation attacks, we assume that the malicious SP has access to both column- and row-wise correlations of Alice's database, which contains (i) correlations between all pairs of attributes (columns), (ii) each individual's membership to the communities and (iii) the statistical relationships of pairs of individuals belonging to  the same community. 
Specifically, the column-wise correlations are characterized by the set of joint distributions among pairs 
of attributes (columns) in the database, i.e., $\mathcal{J} = \{J_{p,q}|p,q\in\mathcal{F},p\neq q\}$. Row-wise correlations, on the other hand, are characterized by the set of statistical relationships between pairs of individuals (rows) in a community. For instance,  $\mathcal{S} = \{s_{ij}^{\mathrm{comm}_c}|i,j\in\mathrm{comm}_c, i\neq j, c\in[1,C]\}$, where $s_{ij}^{\mathrm{comm}_c} = e^{-\mathrm{dist}(\boldsymbol{r}_i,\boldsymbol{r}_j)}$ is the statistical relationship between individuals (data records) $i$ and $j$ in community $\mathrm{comm}_c$ ($\mathrm{dist}(\boldsymbol{r}_i,\boldsymbol{r}_j)$ denotes the Hamming distance between $\boldsymbol{r}_i$ and $\boldsymbol{r}_j$). 
Since the added fingerprint changes some entries in the original  database, which will lead to the change of both joint distributions and statistical relationships, the malicious SP can utilize its auxiliary (publicly available) information about $\mathcal{J}$ and $\mathcal{S}$ to identify the positions of suspicious entries in $\widetilde{\mathbf{R}}$ that are potentially fingerprinted.

\subsection{Column-wise Correlation Attack}
\label{sec:col-wise-attack}


To launch the column-wise correlation attack ($\mathrm{Atk_{col}(\mathcal{J})}$) on  $\widetilde{\mathbf{R}}$, the malicious SP first calculates the empirical joint distributions among pairs of attributes in $\widetilde{\mathbf{R}}$, denoted as $\widetilde{\mathcal{J}}$. 
Then, it compares each joint distribution in $\widetilde{\mathcal{J}}$ (i.e., $\widetilde{J_{p,q}}$) with that in $\mathcal{J}$ (i.e., $J_{p,q}$). For instance, if the absolute  difference of joint probabilities when attribute $p$ takes value $a$ and attribute $q$ takes value $b$ (i.e., $\big|J_{p,q}(a,b)-\widetilde{J_{p,q}}(a,b)\big|$) is higher than a threshold $\tau_{\mathrm{col}}^{\mathrm{Atk}}$, then, the malicious SP queries the row indices 
of the data records in $\widetilde{\mathbf{R}}$ whose attributes $p$ and $q$ take values $a$ and $b$, respectively, and collects the corresponding row indices in a set $\mathcal{I}$, i.e., for the previous example, $\mathcal{I} =  \text{row\ index\ query}\big(\widetilde{\mathbf{R}}.p == a\  \text{and}\ \widetilde{\mathbf{R}}.q == b\big)$  ($\widetilde{\mathbf{R}}.p$ includes   attribute $p$ of all data record in database $\widetilde{\mathbf{R}}$).  
{\color{black}For each row index $i\in\mathcal{I}$, either position $\{i,p\}$ or $\{i,q\}$ (i.e., the row index and attribute tuple) can be   potentially fingerprinted, because they both affect the joint distribution $\widetilde{J_{p,q}}(a,b)$}. Thus, the malicious SP adds each of {\color{black}these tuples}, i.e.,  $\{i,p\}$ and $\{i,q\}, i\in\mathcal{I}$ into   a suspicious position set denoted as $\mathcal{P}$. 

Since a specific suspicious row index $i$ 
can be associated with multiple attributes  in  the suspicious position set $\mathcal{P}$, the suspicious  attribute that is most frequently associated with $i$ is considered to be \textbf{highly suspicious}. 
The malicious SP collects these highly suspicious combinations of row index and attribute in a set $\mathcal{H} = \mathcal{H}\cup\{i,\mathrm{mode}(\mathcal{A}_i)\}$, where $\mathcal{A}_i$ includes all the attributes that are paired with row index $i$ in set $\mathcal{P}$, and $\mathrm{mode}(\mathcal{A}_i)$ returns the most frequent attribute in $\mathcal{A}_i$ (if there is a tie,  the malicious SP randomly chooses one). 
Then, the malicious SP launches the column-wise correlation attack by flipping the LSB of entries in $\widetilde{\mathbf{R}}$ whose positions are in $\mathcal{H}$, i.e., $\widetilde{\mathbf{R}}.(i,p), \forall \{i,p\}\in\mathcal{H}$  {\color{black}($\widetilde{\mathbf{R}}.(i,p)$ represents the value of attribute $p$ for the $i$th data record in   $\widetilde{\mathbf{R}}$)}.

In practice, the malicious SP can launch multiple rounds of $\mathrm{Atk_{col}(\mathcal{J})}$ by iteratively comparing the new joint distributions obtained from the attacked fingerprinted database in the previous round with its prior knowledge $\mathcal{J}$. 
In each round,  a new $\mathcal{H}$ is constructed, but the malicious SP does not flip the highly suspicious positions that have already been flipped in previous rounds. This can be achieved by maintaining and updating an accumulative highly suspicious position set $\mathcal{Z}$. We summarize the steps of conducting $t$ rounds of $\mathrm{Atk_{col}(\mathcal{J})}$ in
Algorithm \ref{Algo:col-wise-attack}. 

\begin{algorithm}
\small
\SetKwInOut{Input}{Input}
\SetKwInOut{Output}{Output}
\Input{Fingerprinted   database $\widetilde{\mathbf{R}}$, malicious SP's prior knowledge on the pairwise joint distributions among attributes, $\mathcal{J}$, and attack rounds $t$.}
\Output{column-wise correlation attacked   DB  $\overline{\mathbf{R}}\Big(\emptyset,\emptyset,\mathrm{Atk_{col}}(\mathcal{J})\Big)$.}

Initialize $cnt = 1$;

Initialize  $\mathcal{Z} = \emptyset$;


\While{$cnt\leq t$}{
Initialize $\mathcal{P} = \emptyset$, $\mathcal{H} = \emptyset$;

Update the empirical joint distributions set $\widetilde{\mathcal{J}}$ using $\widetilde{\mathbf{R}}$;  

      \ForAll{\texttt{$p,q\in\mathcal{F}, p\neq q$}}{
      \ForAll{\texttt{$a\in[0,k_p-1], b\in[0,k_q-1]$}}{\If{$\big|J_{p,q}(a,b)-\widetilde{J_{p,q}}(a,b)\big| \geq \tau_{\mathrm{col}}^{\mathrm{Atk}}$}{
   $\mathcal{I} = \text{row\ index\ query}\big(\widetilde{\mathbf{R}}.p == a\ \text{and}\  \widetilde{\mathbf{R}}.q == b\big)$;  \\
   \ForAll{row index $i\in\mathcal{I}$}{
   
   \If{$\{i,p\}\notin\mathcal{P}$}{$\mathcal{P}=\mathcal{P}\cup \{i,p\}$;}

    \If{$\{i,q\}\notin\mathcal{P}$}{$\mathcal{P}=\mathcal{P}\cup \{i,q\}$;}
    

   }
   }}
      
       }

\ForAll{ row index and attribute tuple  $\{i,p\}\in\mathcal{P}$}{

Collect all attributes   paired with row index $i$ into set $\mathcal{A}_i$

$\mathcal{H} = \mathcal{H} \cup \{i,\rm{mode}(\mathcal{A}_i)\}$;\CommentSty{//the most frequent attribute associate with row index $i$ is recorded in  $\mathcal{H}$.}\\}
     
     \ForAll{highly suspicious row index and attribute tuple $\{i,p\}\in\mathcal{H}$}   {\If{$\{i,p\}\notin\mathcal{Z}$}{Change the LSB of $\widetilde{\mathbf{R}}.(i,p)$;
     
     $\mathcal{Z} = \mathcal{Z}\cup \{i,p\}$;
     

     }}
       
       $cnt = cnt+1;$
       
       
     }

     Return $\overline{\mathbf{R}}\Big(\emptyset,\emptyset,\mathrm{Atk_{col}}(\mathcal{J})\Big) = \widetilde{\mathbf{R}}$;


		\caption{$\mathrm{Atk_{col}}(\mathcal{J})$: Column-wise Correlation     Attack}
	
\label{Algo:col-wise-attack}
\end{algorithm}

{\color{black}
Next, we  show that a malicious SP can   increase its inference power (confidence) about whether  a particular entry in the database is fingerprinted or not by launching $\mathrm{Atk_{col}}(\mathcal{J})$. In Section \ref{sec:eva}, we   experimentally validate this finding using a real-world   database.  
Under  $\mathrm{Atk_{rnd}}$, we denote the malicious SP's confidence that an entry, whose attribute $p$ takes value $a$ in the original database ($\mathbf{R}$), is changed due to the fingerprinting as $\mathrm{Conf_{Atk_{rnd}}}(\frac{1}{\gamma};p,a)$.  
Likewise, under  $\mathrm{Atk_{col}}(\mathcal{J})$, we represent such confidence as $\mathrm{Conf_{Atk_{col}(\mathcal{J})}}(\frac{1}{\gamma};p,a)$. Here, $\gamma\in(0,1)$ is the fingerprinting ratio and we use $\frac{1}{\gamma}$ as the decision parameter to investigate the asymptotic behavior of  the malicious SP's confidence gain, which is defined  as the ratio  $G_{\mathrm{col}}(\frac{1}{\gamma};p,a) = {\mathrm{Conf_{Atk_{col}(\mathcal{J})}}(\frac{1}{\gamma};p,a)}\Big/{\mathrm{Conf_{Atk_{rnd}}}(\frac{1}{\gamma};p,a)}$. Thus, we have the following proposition. 

\begin{prop}\label{prop:atk_col}
By launching $\mathrm{Conf_{Atk_{col}(\mathcal{J})}}$, the malicious SP's   confidence gain
about  an entry, whose attribute $p$ takes value $a$ in $\mathbf{R}$, is fingerprinted can be shown in an  asymptotic manner as
\begin{equation*}
    G_{\mathrm{col}}(\frac{1}{\gamma};p,a) 
    = \Theta\left(\left({1-\prod_{q\in\mathcal{T},q\neq p}\left(\frac{\tau_{\mathrm{col}}^{\mathrm{Atk}}}{\frac{\gamma}{|\mathcal{T}|}{2freq}_{a}^p}\right)^{k_q}}\right)\Bigg/\left({\frac{\gamma}{|\mathcal{T}|}{freq}_{a}^p}\right)\right),
\end{equation*}
where ${freq}_{a}^p$ is the frequency of records with attribute $p$ taking value $a$ in $\mathbf{R}$,   $k_q$ is the number of different values for attribute $q$, and $\Theta(\cdot)$ is the Big-Theta notation. 
\end{prop}
\begin{proof}[Proof Sketch] 
For the vanilla fingerprinting scheme, we have $\mathrm{Conf_{Atk_{rnd}}}(p,a)=\frac{\gamma}{|\mathcal{T}|}{freq}_{a}^p$. When launching the $\mathrm{Atk_{col}}(\mathcal{J})$, the malicious SP will add the corresponding suspicious row index and attribute tuple in $\mathcal{P}$ if $\big|J_{p,q}(a,b)-\widetilde{J_{p,q}}(a,b)\big| \geq \tau_{\mathrm{col}}^{\mathrm{Atk}}$. Thus, we have   $    \mathrm{Conf_{Atk_{col}(\mathcal{J})}}(p,a) = 1-
    \prod_{q\in\mathcal{T},q\neq p}\prod_{b\in[0,k_q-1]} \Pr\Big(\big|J_{p,q}(a,b)-\widetilde{J_{p,q}}(a,b)\big| < \tau_{\mathrm{col}}^{\mathrm{Atk}}\Big)$. Since the inserted fingerprint will cause  $J_{p,q}(a,b)-\widetilde{J_{p,q}}(a,b)$ vary in the range of $\Big[-\frac{\gamma}{|\mathcal{T}|}2freq_{a,b}^{p,q},\frac{\gamma}{|\mathcal{T}|}2 freq_{a,b}^{p,q}\Big]$, where $freq_{a,b}^{p,q}$ is the frequency of entries whose attributes $p$ and $q$ take values $a$ and $b$ in $\mathbf{R}$. 
    Then, $|J_{p,q}(a,b)-\widetilde{J_{p,q}}(a,b)|$ can be shown as  a random variable attributed to an uniform distribution in the support of $\Big[0,\frac{\gamma}{|\mathcal{T}|}2 freq_{a,b}^{p,q}\Big]$, which leads to  $\mathrm{Conf_{Atk_{col}(\mathcal{J})}}(p,a) = 1- \prod_{q\in\mathcal{T},q\neq p}\prod_{b\in[0,k_q-1]} {\tau_{\mathrm{col}}^{\mathrm{Atk}}}\Big/\Big({\frac{\gamma}{|\mathcal{T}|}{2freq}_{a}^p} \Big)$. By applying   arithmetic-geometric mean inequality along with the fact   $\sum_{b\in[0,k_q-1]}freq_{a,b}^{p,q}  = freq_{a}^{p}, \forall q\neq p$, we can complete the proof.
\end{proof}
\begin{remark}
We aim at presenting a generic confidence gain achieved from $\mathrm{Atk_{col}}(\mathcal{J})$, thus we consider the potential fingerprinted entries in the suspicious set $\mathcal{P}$ instead of the highly suspicious set $\mathcal{H}$. In practice, the generation process of  $\mathcal{H}$ from $\mathcal{P}$ heavily depends on the data distribution in the considered databases. 
\end{remark}

}

\subsection{Row-wise Correlation Attack}
\label{sec:row-wise-atack}


Since the malicious SP has access to both individuals' memberships to   communities and row-wise correlations, i.e.,  $\mathcal{S}$, 
after receiving the fingerprinted database, it can compute a new set of statistical relationships 
among pairs of individuals in each of the communities using $\widetilde{\mathbf{R}}$, 
i.e.,  $\widetilde{\mathcal{S}} =  \{\widetilde{s_{ij}}^{\mathrm{comm}_c}|i,j\in\mathrm{comm}_c, i\neq j, c\in[1,C]\}$, where $\widetilde{s_{ij}}^{\mathrm{comm}_c} = e^{-\mathrm{dist}(\widetilde{\boldsymbol{r}_i},\widetilde{\boldsymbol{r}_j})}$ is the statistical relationship between the $i$th and $j$th data records (i.e., $\widetilde{\boldsymbol{r}_i}$ and $\widetilde{\boldsymbol{r}_j}$) in $\widetilde{\mathbf{R}}$. 
Then, to conduct $\mathrm{Atk_{row}}(\mathcal{S})$, the malicious SP flips the LSBs of all attributes of a   data record $\boldsymbol{r}_i$, 
if the cumulative absolute difference of its  statistical relationships with respect to other records in the same community exceeds a predetermined threshold $\tau_{\mathrm{row}}^{\mathrm{Atk}}$ after fingerprinting, i.e., $\sum_{j\neq i}^{n_c} \   \big|s_{ij}^{\mathrm{comm}_c}-\widetilde{s_{ij}}^{\mathrm{comm}_c}\big| \geq \tau_{\mathrm{row}}^{\mathrm{Atk}}, i,j\in\mathrm{comm}_c$.   {\color{black} The rationale behind this is because the row-wise statistical information $\mathcal{S}$ is calculated using  the entire data records, instead of individual entries between the rows. Although, this represents the strongest row-wise attack as it changes all the entries of a given data record, in practice, $\mathrm{Atk_{row}(\mathcal{S})}$ changes only a limited number of data records, as will be shown in Section \ref{sec: int_attack_census}.} We summarize the steps to launch $\mathrm{Atk_{row}}(\mathcal{S})$ on $\widetilde{\mathbf{R}}$ in Algorithm \ref{algo:row-wise-attack}. 

\begin{algorithm}
\small
\SetKwInOut{Input}{Input}
\SetKwInOut{Output}{Output}
\Input{Fingerprinted database, $\widetilde{\mathbf{R}}$, malicious SP's prior knowledge on the row-wise correlations  $\mathcal{S}$ and individuals' affiliation to the $C$ communities.}
\Output{$\overline{\mathbf{R}}\Big(\emptyset,\mathrm{Atk_{row}}(\mathcal{S}),\emptyset,\Big)$.}


Obtain  the new set of  pairwise statistical relationships among individuals in each community   from  $\widetilde{\mathbf{R}}$, i.e., $\widetilde{\mathcal{S}}$;

\ForAll{$\mathrm{comm}_c, c\in[1,C]$}{\ForAll{individual $i\in \mathrm{comm}_c$}{\If{$\sum_{j\neq i}^{n_c} \   \big|s_{ij}^{\mathrm{comm}_c}-\widetilde{s_{ij}}^{\mathrm{comm}_c}\big| \geq \tau_{\mathrm{row}}^{\mathrm{Atk}}$}{Flip the LSBs of all attributes of $\boldsymbol{r}_i$ in $\widetilde{\mathbf{R}}$;}

}}

 Return $\overline{\mathbf{R}}\Big(\emptyset,\mathrm{Atk_{row}}(\mathcal{S}),\emptyset,\Big) = \widetilde{\mathbf{R}}$.

		\caption{$\mathrm{Atk_{row}}(\mathcal{S})$: Row-wise correlation  attack}
\label{algo:row-wise-attack}
\end{algorithm}


We analyze the impact of $\mathrm{Atk_{row}}(\mathcal{S})$ by denoting the malicious SP's confidence that an entry ($\boldsymbol{r}_i$) is fingerprinted as $\mathrm{Conf}_{\mathrm{Atk_{rnd}}}(\frac{1}{\gamma};\boldsymbol{r}_i)$ and $\mathrm{Conf}_{\mathrm{Atk_{row}}(\mathcal{S})}(\frac{1}{\gamma};\boldsymbol{r}_i)$, under $\mathrm{Atk_{rnd}}$ and $\mathrm{Atk_{row}}(\mathcal{S})$, respectively. Then, 
the confidence  gain of the malicious SP is  $G_{\mathrm{row}}(\frac{1}{\gamma};\boldsymbol{r}_i) = \frac{\mathrm{Conf}_{\mathrm{Atk_{row}}(\mathcal{S})}(\frac{1}{\gamma};\boldsymbol{r}_i)}{\mathrm{Conf}_{\mathrm{Atk_{rnd}}}(\frac{1}{\gamma};\boldsymbol{r}_i)}$, which is calculated in the following proposition.

\begin{prop}\label{prop:atk-row}
By launching $\mathrm{Conf_{Atk_{row}(\mathcal{S})}}$, the malicious SP's maximum confidence gain  about  an entry   in $\mathbf{R}$ is fingerprinted can be shown asymptotically as  
\begin{equation*}
    G_{\mathrm{row}}(\frac{1}{\gamma};\boldsymbol{r}_i)  = \Theta\vast( \Bigg(1-\sum_{j=0}^{\floor*{\tau_{\mathrm{row}}^{\mathrm{Atk}}}}{n_c-1 \choose j}(2\gamma-\gamma^2)^j(1-\gamma)^{2(n_c-1-j)}\Bigg)  \Big/ \gamma \vast).
\end{equation*}
\end{prop}
\begin{proof}[Proof Sketch]
Clearly, $\mathrm{Conf}_{\mathrm{Atk_{rnd}}}(\boldsymbol{r}_i) = \gamma$. 
According to Algorithm \ref{algo:row-wise-attack},  $\mathrm{Conf}_{\mathrm{Atk_{row}}(\mathcal{S})}(\frac{1}{\gamma};\boldsymbol{r}_i) = \Pr(\sum_{j\neq i}^{n_c} \   \big|e^{-\mathrm{dist}(\boldsymbol{r}_i,\boldsymbol{r}_j)}-e^{-\mathrm{dist}(\widetilde{\boldsymbol{r}_i},\widetilde{\boldsymbol{r}_j})}\big| \geq \tau_{\mathrm{row}}^{\mathrm{Atk}})\stackrel{*}\approx \Pr(\sum_{j\neq i}^{n_c} \   \big|\mathrm{dist}(\boldsymbol{r}_i,\boldsymbol{r}_j)-\mathrm{dist}(\widetilde{\boldsymbol{r}_i},\widetilde{\boldsymbol{r}_j})\big| \geq \tau_{\mathrm{row}}^{\mathrm{Atk}})$, where $*$ is due to the Taylor approximation and the assumption that the distance between individuals in the same community is   small. Then,  $|\mathrm{dist}(\boldsymbol{r}_i,\boldsymbol{r}_j)-\mathrm{dist}(\widetilde{\boldsymbol{r}_i},\widetilde{\boldsymbol{r}_j})\big|$ can be shown as a Bernoulli  random variable, which is $0$ with probability $(1-\gamma)^2$, and is nonzero with probability $2\gamma-\gamma^2$. Since the summation of Bernoulli random variable is attributed to binomial distribution, we can finish the proof.
\end{proof}
\begin{remark}
In the above analysis, we   ignored the scenario where $|\mathrm{dist}(\boldsymbol{r}_i,\boldsymbol{r}_j)-\mathrm{dist}(\widetilde{\boldsymbol{r}_i},\widetilde{\boldsymbol{r}_j})\big|$ is $2$ with probability $\gamma^2$ to avoid extra heavy notations. In the experiments, we set $\gamma=1/35$, thus, $\gamma^2$ is negligible. 
\end{remark}





\subsection{Integrated Correlation Attack}
\label{sec:effect_corr_attack}

In practice, the malicious SP will apply $\mathrm{Atk_{row}}(\mathcal{S})$ followed by $\mathrm{Atk_{col}}(\mathcal{J})$ if it launches the integrated correlation attack. {\color{black}This is because  (i) $\mathrm{Atk_{row}}(\mathcal{S})$ is computationally light and  modifies significantly  less entries in $\widetilde{\mathbf{R}}(\mathrm{FP},\emptyset,\emptyset)$ compared to $\mathrm{Atk_{col}}(\mathcal{J})$ (as we will show in Section~\ref{sec: int_attack_census}). (ii) If $\mathrm{Atk_{col}}(\mathcal{J})$ is applied first, it will change   the row-wise correlations ($\mathrm{P_{row}}$)  significantly, yet, if $\mathrm{Atk_{row}}(\mathcal{S})$ is applied first, it only has a small impact on the column-wise correlations $\mathrm{P_{col}}$ (as we will also show in Section~\ref{sec: int_attack_census}).}  Algorithm \ref{algo:int-attack}    summarizes the major steps of this integrated attack. Note that, in practice,  there is no  minimum distribution difference requirement  to perform the proposed attacks, because a malicious SP can always reduce the value of  $\tau_{\mathrm{col}}^{\mathrm{Atk}}$ and $\tau_{\mathrm{row}}^{\mathrm{Atk}}$ to obtain more potentially fingerprinted entries.

\begin{algorithm}
\small
\SetKwInOut{Input}{Input}
\SetKwInOut{Output}{Output}
\Input{Fingerprinted database, $\widetilde{\mathbf{R}}$, malicious SP's prior knowledge on the row-wise correlations  $\mathcal{S}$, individuals' affiliation to the $C$ communities, and column-wise correlations  $\mathcal{J}$.}
\Output{$\overline{\mathbf{R}}\Big(\emptyset,\mathrm{Atk_{row}}(\mathcal{S}),\mathrm{Atk_{col}}(\mathcal{J})\Big)$.}


Launch row-wise correlation attack   $\mathrm{Atk_{row}}(\mathcal{S})$ on $\widetilde{\mathbf{R}}$ using Algorithm \ref{algo:row-wise-attack}, and obtain $\overline{\mathbf{R}}\Big(\emptyset,\mathrm{Atk_{row}}(\mathcal{S}),\emptyset\Big)$;

Launch  column-wise correlation attack  $\mathrm{Atk_{col}}(\mathcal{J})$ on $\overline{\mathbf{R}}\Big(\emptyset,\mathrm{Atk_{row}}(\mathcal{S}),\emptyset\Big)$ using Algorithm \ref{Algo:col-wise-attack}, obtain and return   $\overline{\mathbf{R}}\Big(\emptyset,\mathrm{Atk_{row}}(\mathcal{S}),\mathrm{Atk_{col}}(\mathcal{J})\Big)$.

\caption{Integrated correlation attack}
\label{algo:int-attack}
\end{algorithm}



By taking advantage of the correlation models, the identified attacks (in Sections \ref{sec:col-wise-attack} and \ref{sec:row-wise-atack}) are   more powerful than the traditional random bit flipping attack  (in Section \ref{sec:threat_models}). As we will show in Section~\ref{sec: int_attack_census}, to effectively distort the added fingerprint and cause Alice to accuse innocent SPs with high probability, a malicious SP only needs to change a small fraction of entries in the fingerprinted database if it conducts the correlation attacks  on $\widetilde{\mathbf{R}}(\mathrm{FP},\emptyset,\emptyset)$. 
In contrast, to achieve a similar attack performance, the random bit flipping attack   needs to change  more than $80\%$ of the entries in $\widetilde{\mathbf{R}}(\mathrm{FP},\emptyset,\emptyset)$, which results in a significant loss in     database utility. Thus, the   correlation attacks not only distort the inserted fingerprint but they also maintain a high utility for the pirated database.  

Due to the identified vulnerability of existing fingerprinting schemes for relations against correlation attacks, it is critical to develop defense mechanisms that can mitigate these attacks. In the next section, we discuss how to develop robust fingerprinting techniques against both column- and row-wise correlation attacks.

\section{Robust Fingerprinting Against Identified Correlation Attacks}
\label{sec:robust_fp}

Now, we propose robust fingerprinting schemes against the identified correlation attacks that can serve as post-processing steps for any off-the-shelf (vanilla) fingerprinting schemes.  
To provide robustness against column- and row-wise correlation attack, i.e.,  $\mathrm{Atk_{col}(\mathcal{J})}$ and   $\mathrm{Atk_{row}(\mathcal{S})}$, the database owner (Alice) utilizes her prior knowledge $\mathcal{J}'$ and $\mathcal{S}'$ as the reference column-wise joint distributions and statistical relationships, respectively. We will show that to implement the proposed mitigation techniques, Alice needs to change only a few entries (e.g., less than $3\%$) in $\widetilde{\mathbf{R}}(\mathrm{FP},\emptyset,\emptyset)$, 
such that the post-processed fingerprinted database has    column-wise correlation    close to $\mathcal{J}'$ and row-wise correlation     far from  $\mathcal{S}'$.


\subsection{Robust Fingerprinting Against Column-wise Correlation Attack}
\label{sec:col_defense}


\subsubsection{Mitigation via mass transportation} 
To make a vanilla fingerprinting scheme robust against column-wise correlation attack, the main goal of the proposed technique   $\mathrm{Dfs_{col}}(\mathcal{J}')$ is to transform  $\widetilde{\mathbf{R}}(\mathrm{FP},\emptyset,\emptyset)$ to have column-wise joint distributions close to  the reference joint distributions in $\mathcal{J}'$. We develop  $\mathrm{Dfs_{col}}(\mathcal{J}')$ using ``optimal transportation''~\cite{courty2016optimal}, which moves the probability  mass of the marginal distribution of each attribute in $\widetilde{\mathbf{R}}(\mathrm{FP},\emptyset,\emptyset)$ to resemble the distribution obtained from the marginalization of each reference joint distribution in $\mathcal{J}'$. Then, the optimal transportation plan is used to change the entries in each attribute of  $\widetilde{\mathbf{R}}(\mathrm{FP},\emptyset,\emptyset)$ to obtain $\widetilde{\mathbf{R}}\big(\mathrm{FP},\emptyset,\mathrm{Dfs_{col}}(\mathcal{J}')\big)$. 
While doing this, the new empirical joint distributions calculated from     $\widetilde{\mathbf{R}}\big(\mathrm{FP},\emptyset,\mathrm{Dfs_{col}}(\mathcal{J}')\big)$ also become close to the ones in  $\mathcal{J}'$. 


In particular, for a specific attribute (column) $p$, we denote its marginal distribution obtained from the (vanilla) fingerprinted database as $\Pr(C_{\widetilde{p}})$, and that obtained from the  marginalization of a reference $J'_{p,q}$ distribution in $\mathcal{J}'$ as $\Pr(C_{p'}) = J'_{p,q}\mathbf{1}^T$ ($q$ can be any attribute that is different from $p$, because the marginalization with respect to $p$ using different $J'_{p,q}$ will lead to the identical marginal distribution of $p$).
To move the mass of $\Pr(C_{\widetilde{p}})$ to resemble $\Pr(C_{p'})$, we need to find another joint distribution (i.e., the mass transportation plan)  $G_{\widetilde{p},p'}\in\R^{k_p\times k_p}$ ($k_p$ is the number of possible values that attribute $p$ can take), whose marginal distributions are identical to $\Pr(C_{\widetilde{p}})$ and $\Pr(C_{p'})$. 
Let $a$ and $b$ be two distinct values that attribute $p$ can take ($a,b\in[0,k_p-1]$). Then,  $G_{\widetilde{p},p'}(a,b)$ indicates that the database owner should change $G_{\widetilde{p},p'}(a,b)$ percentage of entries in $\widetilde{\mathbf{R}}(\mathrm{FP},\emptyset,\emptyset)$ whose attribute $p$ takes value $a$ (i.e., $p=a$) to value $b$ (i.e., change them to make $p=b$), so as to  make $\Pr(C_{\widetilde{p}})$ close to $\Pr(C_{p'})$. In practice, such a transportation plan can be obtained by solving a regularized optimal transportation problem, i.e., the entropy regularized Sinkhorn distance minimization  \cite{cuturi2013sinkhorn} as follows:  
\begin{equation}
\label{ot_defense}
\begin{aligned}
&d\Big(\Pr(C_{\widetilde{p}}),\Pr(C_{p'}),\lambda_p\Big) \\
&= \min_{G_{\widetilde{p},p'}\in\mathcal{G}\big(\Pr(C_{\widetilde{p}}),\Pr(C_{p'})\big)}<G_{\widetilde{p},p'},\Theta_{\widetilde{p},p'}>_F-\frac{H(G_{\widetilde{p},p'})}{\lambda_{p}},
\end{aligned}
\end{equation}
where     $\mathcal{G}\big(\Pr(C_{\widetilde{p}}),\Pr(C_{p'})\big) = \big\{G\in\R^{k_p\times k_p}\big|G\mathbf{1} = \Pr(C_{\widetilde{p}}),G^T\mathbf{1}  = \Pr(C_{p'})\big\}$ is the set of all  joint probability distributions whose marginal distributions are the probability mass functions   of  $\Pr(C_{\widetilde{p}})$ and $\Pr(C_{p'})$. $<\cdot,\cdot>_F$ denotes the Frobenius inner product 
of two matrices with the same size. Also, $\Theta_{\widetilde{p},p'}$ is the transportation cost matrix and $\Theta_{\widetilde{p},p'}(a,b)>0$ represents the cost to move a unit percentage of mass from $\Pr(C_{\widetilde{p}} = a)$ to $\Pr(C_{\widetilde{p}} = b)$. Finally, $H(G_{\widetilde{p},p'}) = -<G_{\widetilde{p},p'},\log G_{\widetilde{p},p'}>_F$ calculates  the information entropy of $G_{\widetilde{p},p'}$ and $\lambda_p>0$ is a tuning parameter. In practice, (\ref{ot_defense}) can be solved by iteratively rescaling rows and columns of the initialized $G_{\widetilde{p},p'}$ to have desired marginal distributions.  The obtained $G_{\widetilde{p},p'}$ is more heterogeneous for larger values of $\lambda_p$. This suggests that the transportation plan 
tends to move the mass of 
$\Pr(C_{\widetilde{p}} = a)$ to the adjacent instances, i.e, $b=a-1$ or $b=a+1$. In contrast, the obtained $G_{\widetilde{p},p'}$ is more homogeneous for smaller values of $\lambda_p$, which suggests that the transportation plan tends to move the mass of 
$\Pr(C_{\widetilde{p}} = a)$ to all other instances.  {\color{black}A  homogeneous plan makes   $\Pr(C_{\widetilde{p}})$ much closer to $\Pr(C_{p'})$ after the mass transportation, but it causes more data entries to be changed, and results in a higher decrease in the database utility. On the other hand, a heterogeneous plan changes less data entries by tolerating a larger difference between  $\Pr(C_{\widetilde{p}})$ and $\Pr(C_{p'})$ after the mass transportation. In the evaluation (in Section~\ref{sec:eva}), we will try different values of  $\lambda_p$ to strike a balance between the mitigation performance and data utility.}


\subsubsection{A toy example on mass transportation}
\label{sec:toy_example}
To illustrate $\mathrm{Dfs_{col}}(\mathcal{J}')$ via mass transportation of $\Pr(C_{\widetilde{p}})$ to resemble $\Pr(C_{p'})$, we use a pair of discrete probability distributions shown in Figure \ref{fig:OT_toy_example}(a) as an example, and demonstrate the transportation plans obtained by solving (\ref{ot_defense}) when $\lambda_p=500$ and $\lambda_p=50$ in Figures \ref{fig:OT_toy_example}(b) and (c), respectively.
In Figure \ref{fig:OT_toy_example}(b), we have a heterogeneous  $G_{\widetilde{p},p'}$, which often  moves the mass to adjacent instances, e.g., the mass of $\Pr(C_{\widetilde{p}}=0)$ is divided into 3 parts and a larger portion of mass is moved to $\Pr(C_{\widetilde{p}}=1)$.    $G_{\widetilde{p},p'}(0,1) = 0.157$,  thus $0.157$ mass  of $\Pr(C_{\widetilde{p}}=0)$ is moved to $\Pr(C_{\widetilde{p}}=1)$.  In Figure~\ref{fig:OT_toy_example}(c), we obtain a homogeneous $G_{\widetilde{p},p'}$, which distributes the mass to many other instances. For example, the mass of $\Pr(C_{\widetilde{p}}=0)$ is divided into 5 parts and 4 of them are  moved to $\Pr(C_{\widetilde{p}}=1)$, $\Pr(C_{\widetilde{p}}=2)$, $\Pr(C_{\widetilde{p}}=3)$, and $\Pr(C_{\widetilde{p}}=5)$.

\begin{figure}[htb]
  \begin{center}
     \includegraphics[width= 0.7\columnwidth]{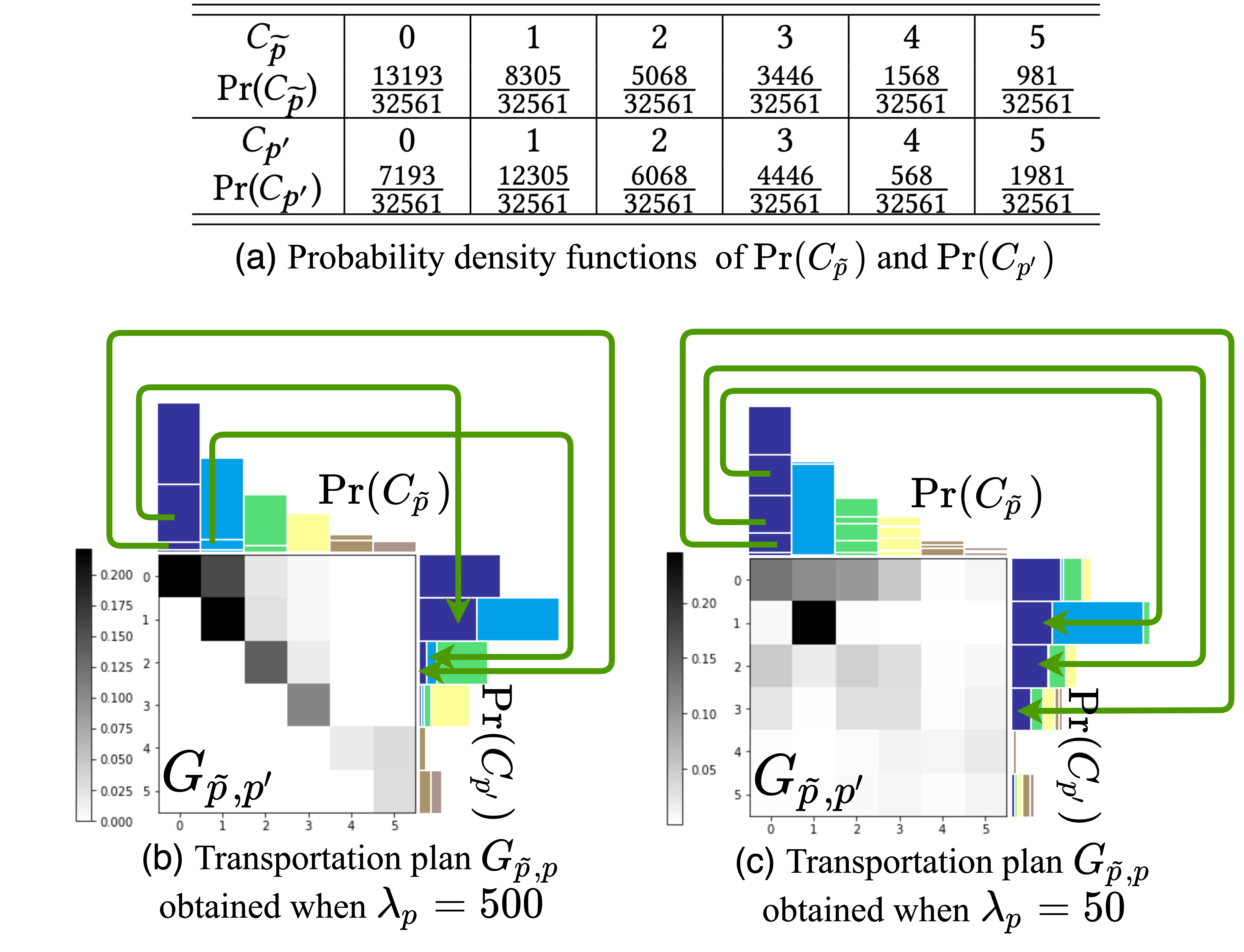}
      \end{center}
  \caption{\label{fig:OT_toy_example} Visualization of mass transportation plans obtained by solving (\ref{ot_defense}) using different $\lambda_p$ values to move   mass of $\Pr(C_{\widetilde{p}})$ to resemble $\Pr(C_{p'})$. (a) example discrete probability distributions. (b) if $\lambda_p=500$, we achieve a heterogeneous plan, which tolerates more   difference between  $\Pr(C_{\widetilde{p}})$ and $\Pr(C_{p'})$ after the mass transportation. 
  (c) if $\lambda_p=50$, we achieve a homogeneous plan. which makes    $\Pr(C_{\widetilde{p}})$ more closer to $\Pr(C_{p'})$ after the mass transportation.
}
\end{figure}

\subsubsection{Algorithm description} In the following, we formally describe the procedure of $\mathrm{Dfs_{col}}(\mathcal{J}')$. After Alice generates $\widetilde{\mathbf{R}}(\mathrm{FP},\emptyset,\emptyset)$ using the vanilla fingerprinting scheme, she   evaluates the new  joint distributions of all pairs of attributes, i.e., $\widetilde{J_{p,q}}, p,q\in\mathcal{F},p\neq q$, and compares them  with the reference joint distributions   $\mathcal{J}_{p,q}',  p,q\in\mathcal{F},p\neq q$.
If the discrepancy between a particular pair of joint distributions exceeds a predetermined threshold, i.e., 
 $||\widetilde{J_{p,q}}-J'_{p,q}||_F\geq \tau_{\mathrm{col}}^{\mathrm{Dfs}}$, Alice records both attributes $p$ and $q$ in a set $\mathcal{Q}$. 
For all the attributes in $\mathcal{Q}$, Alice obtains $\Pr(C_{\widetilde{p}})$ from $\widetilde{\mathbf{R}}(\mathrm{FP},\emptyset,\emptyset).p$ 
and calculates $\Pr(C_{p'}) = J'_{p,q}\mathbf{1}^T$. Next, she gets the optimal transportation plan for attribute $p$  
by solving (\ref{ot_defense}). Then, she changes the instances  of $\widetilde{\mathbf{R}}(\mathrm{FP},\emptyset,\emptyset).p$ to other instances by following the 
transportation moves suggested by   $G_{\widetilde{p},p'}$, i.e., given $G_{\widetilde{p},p'}(a,b)$,  Alice randomly samples $G_{\widetilde{p},p'}(a,b)$ fraction of entries (excluding the fingerprinted entries) whose attribute $p$ takes value $a$ and changes them to $b$.   We summarize the procedure of $\mathrm{Dfs_{col}}(\mathcal{J}')$ in Algorithm \ref{algo:ot_defense}, where lines \ref{line:ot_start}-\ref{line:ot_end} solves (\ref{ot_defense}) to obtain the optimal mass transportation plan for attribute $p$, and lines \ref{line:move_start}-\ref{line:move_end} change the values of entries in $\widetilde{\mathbf{R}}(\mathrm{FP},\emptyset,\emptyset).p$ according to $G_{\widetilde{p},p'}$.

\begin{algorithm}
\small
\SetKwInOut{Input}{Input}
\SetKwInOut{Output}{Output}
\Input{Vanilla fingerprinted database  $\widetilde{\mathbf{R}}(\mathrm{FP},\emptyset,\emptyset)$, locations  of entries changed by the vanilla fingerprinting  scheme, and Alice's prior knowledge on the   joint distributions of the pairwise attributes, i.e.,  $\mathcal{J}'$.}
\Output{$\widetilde{\mathbf{R}}\Big(\mathrm{FP},\emptyset,\mathrm{Dfs_{col}}(\mathcal{J}')\Big)$.}


Initialize $\mathcal{Q}=\emptyset$;

Obtain the empirical joint distributions set $\widetilde{\mathcal{J}}$ using $\widetilde{\mathbf{R}}(\mathrm{FP},\emptyset,\emptyset)$;

\ForAll{$p,q\in\mathcal{F},p\neq q$}{\If{$||J'_{p,q}-\widetilde{J}_{p,q}||_F>\tau_{\mathrm{col}}^{\mathrm{Dfs}}$}{$\mathcal{Q} = \mathcal{Q}\cup p \cup q$;}}

\ForAll{$p\in\mathcal{Q}$}{

Initialize the mass movement cost matrix $\Theta_{\widetilde{p},p'}$ and tuning parameter $\lambda_p$;\label{line:ot_start}

Obtain empirical marginal  distribution $\Pr(C_{\widetilde{p}})$ from $\widetilde{\mathbf{R}}(\mathrm{FP},\emptyset,\emptyset).p$;

       Initialize $G_{\widetilde{p},p'} = e^{-\lambda_{p}\Theta_{\widetilde{p},p'}}$; 
        
\While{not converge}{
Scale the rows of $G_{\widetilde{p},p'}$ to make  the rows sum to the marginal distribution $\Pr(C_{\widetilde{p}})$;

Scale the columns of $G_{\widetilde{p},p'}$ to make the columns sum to the marginal distribution $\Pr(C_{p'})$;
}
 
 \ForAll{$a\in[0,k_p-1]$\label{line:move_start}}{\ForAll{$b\in [0,k_p-1], b\neq a$}{Sample $G_{\widetilde{p},p'}(a,b)$ percentage of entries from $\widetilde{\mathbf{R}}(\mathrm{FP},\emptyset,\emptyset).p$ (excluding the vanilla  fingerprinted entries) whose attribute $p$ takes value $a$, and change their value to $b$;}
  }
 }
 
 Return $\widetilde{\mathbf{R}}\Big(\mathrm{FP},\emptyset,\mathrm{Dfs_{col}}(\mathcal{J}')\Big)$.

		\caption{$\mathrm{Dfs_{col}}(\mathcal{J}')$: defense against column-wise correlation   attack.}
\label{algo:ot_defense}
\end{algorithm}

\subsubsection{Design details of $\mathrm{Dfs_{col}}(\mathcal{J}')$.} 
\label{sec:supp_design_detail_of_dfs_col}

We do not apply the optimal transportation technique to directly move the mass of the joint distributions obtained from $\widetilde{\mathbf{R}}(\mathrm{FP},\emptyset,\emptyset)$ to resemble the joint distributions in    $\mathcal{J}'$. 
One  reason is that, to do so, the database owner (Alice) needs to solve (\ref{ot_defense}) for  $\frac{|\mathcal{F}|(|\mathcal{F}|-1)}{2}$ joint distributions. This is computationally expensive if the  database includes a large number of attributes. Thus, by considering the mass transportation in marginal distributions, the developed mitigation technique becomes more efficient.  
Furthermore, by only considering the marginal distributions, Alice can arrange $\widetilde{\mathbf{R}}\Big(\mathrm{FP},\emptyset,\mathrm{Dfs_{col}}(\mathcal{J}')\Big)$ to have  Pearson's  correlations among attribute pairs that are close to those obtained from $\overline{\mathbf{R}}\Big(\emptyset,\emptyset,\mathrm{Atk_{col}}(\mathcal{J})\Big)$ if $\mathcal{J}'$ is close to $\mathcal{J}$. 
For instance, denote the Pearson's correlation between attributes $p$ and $q$ calculated  from $\widetilde{\mathbf{R}}\Big(\mathrm{FP},\emptyset,\mathrm{Dfs_{col}}(\mathcal{J}')\Big)$ and $\overline{\mathbf{R}}\Big(\emptyset,\emptyset,\mathrm{Atk_{col}}(\mathcal{J})\Big)$ as $\rho_{p',q'}$ and $\rho_{p,q}$, respectively. Then, we have
    $\rho_{p',q'} =  \frac{\sum_{a,b}(a -\mu_{C_{p'}})(b -\mu_{C_{q'}})J_{p,q}'(a,b)}{\sigma_{C_{p'}}\sigma_{C_{q'}}}$, 
where $\mu_{C_{p'}}$ (or $\mu_{C_{q'}}$) 
and $\sigma_{C_{p'}}$ (or $\sigma_{C_{q'}}$) is the expected value   and the standard deviation  of attribute $p$ (or $q$) obtained after applying the vanilla fingerprinting scheme followed by $\mathrm{Dfs_{col}}(\mathcal{J}')$, respectively. Also, $J_{p,q}'(a,b)$ is the database owner's prior knowledge on the joint probability distribution of attribute $p$ taking value $a$ and attribute $q$ taking value $b$. Likewise, 
    $\rho_{p,q} =  \frac{\sum_{a,b}(a -\mu_{C_{p}})(b -\mu_{C_{q}})J_{p,q}(a,b)}{\sigma_{C_{p}}\sigma_{C_{q}}}$, 
where $\mu_{C_{p}}$ (or $\mu_{C_{q}}$) and 
$\sigma_{C_{p}}$ (or $\sigma_{C_{q}}$) is the expected value  and the standard deviation of attribute $p$ (or $q$) in $\overline{\mathbf{R}}\Big(\emptyset,\emptyset,\mathrm{Atk_{col}}(\mathcal{J})\Big)$. Also, $J_{p,q}(a,b)$ is the malicious SP's  prior knowledge on the joint probability distribution of attribute $p$ taking value $a$ and attribute $q$ taking value $b$. If $\mathcal{J}'$ is close to $\mathcal{J}$, then $\mu_{C_{p'}}$ (or $\mu_{C_{q'}}$) is also close to $\mu_{C_{p}}$ (or $\mu_{C_{q}}$), because of the marginalization of the similar joint distributions. Similar discussion also holds for $\sigma_{C_{p'}}$ (or $\sigma_{C_{q'}}$) and $\sigma_{C_{p}}$ (or $\sigma_{C_{q}}$). As a result, $\rho_{p',q'}$ also becomes close to $\rho_{p,q}$, which improves the robustness of the fingerprint (against correlation attacks by a malicious SP), and hence prevents a malicious SP from distorting the potentially fingerprinted positions.

\subsection{Robust Fingerprinting Against Row-wise Correlation Attack}
\label{sec:row-defense}


To make a vanilla fingerprinting scheme also robust against row-wise correlation attack (in Section \ref{sec:row-wise-atack}), we develop another mitigation technique, i.e.,  $\mathrm{Dfs_{row}}(\mathcal{S}')$.
The main goal of $\mathrm{Dfs_{row}}(\mathcal{S}')$ is to avoid a malicious SP from distorting the fingerprint due to discrepancies in the expected statistical relationships between data records. 
Different from the design principle of $\mathrm{Dfs_{col}}(\mathcal{J}')$, which makes the newly obtained  joint distributions resemble the prior knowledge, we design $\mathrm{Dfs_{row}}(\mathcal{S}')$ by changing selected entries of non-fingerprinted data records to  make the newly obtained  statistical relationships as far away from Alice's prior knowledge $\mathcal{S}'$ as possible. 
This is because the row-wise correlation attack usually changes limited number of entries in the vanilla fingerprinted database (as we validate in Section \ref{sec: int_attack_census}), thus, to make the newly obtained  statistical relationships   resemble $\mathcal{S}'$, one needs to change all non-fingerprinted data records and this will significantly compromise the database utility. 
{\color{black}Instead, by making the new statistical relationships far away from her prior knowledge, Alice can make additional (non-fingerprinted) data records that have cumulative absolute difference (with respect to the other records in the same community) exceeding a predetermined threshold. As a result, when launching $\mathrm{Atk_{row}(\mathcal{J})}$, the malicious SP will identify wrong data records ($\boldsymbol{r}_i$), which causes $\sum_{j\neq i}^{n_c} \   \big|s_{ij}^{\mathrm{comm}_c}-\widetilde{s_{ij}}^{\mathrm{comm}_c}\big| \geq \tau_{\mathrm{row}}^{\mathrm{Atk}}$, and hence change the non-fingerprinted records.}

In $\mathrm{Dfs_{row}}(\mathcal{S}')$, Alice selects a subset of  non-fingerprinted data records in a community $c$, i.e., $\mathcal{E}_c\subset\mathrm{comm}_c, c\in[1,C]$, and changes their value to   $\widetilde{\widetilde{\boldsymbol{r}_i}}, i\in\mathcal{E}_c$, such that the cumulative absolute  difference  between statistical relationships in her prior knowledge and those obtained from the   fingerprinted database  achieves the maximum difference   after   applying   $\mathrm{Dfs_{row}}(\mathcal{S}')$. This can be formulated  as  the following optimization problem:

\begin{equation}
\begin{aligned}
\max_{\mathcal{E}_c,\widetilde{\widetilde{\boldsymbol{r}_i}}} \quad &  d(\mathcal{E}_c)=   \Big| \sum_{j\in  \mathrm{comm}_c/\mathcal{E}_c}\sum_{i\in\mathcal{E}_c} \   \Big|{s'_{ij}}^{\mathrm{comm}_c}-\widetilde{\widetilde{s_{ij}}}^{\mathrm{comm}_c}\Big| \\
& \qquad\quad -\sum_{j\in  \mathrm{comm}_c/\mathcal{E}_c}\sum_{i\in\mathcal{E}_c} \   \Big|{s'_{ij}}^{\mathrm{comm}_c}- \widetilde{s_{ij}}^{\mathrm{comm}_c}\Big| \Big| \\
\textrm{s.t.} \quad & \mathcal{E}_c\subset \mathrm{comm}_c/\mathcal{Q}_c,\\
   & \widetilde{\widetilde{s_{ij}}}^{\mathrm{comm}_c} = e^{-\mathrm{dist}(\widetilde{\widetilde{\boldsymbol{r}_i}},\boldsymbol{r}_j)}, i\in\mathcal{E}_c,j \in\mathrm{comm}_c/\mathcal{E}_c,\\
   & \widetilde{\widetilde{\boldsymbol{r}_i}} = \mathrm{value\  change}(\widetilde{\boldsymbol{r}_i}), i \in\mathcal{E}_c,\\
    &  |\mathcal{E}_c|\leq \left \lceil{n_c\gamma}\right \rceil,
\end{aligned}
\label{eq:row_defense}
\end{equation}
$\forall c\in[1,C]$. $\mathcal{Q}_c$ is the set of fingerprinted records in community $c$, ${s_{ij}'}^{\mathrm{comm}_c}$ denotes Alice's prior knowledge on  the statistical relationship between individuals $i$ and $j$ in community $c$,  $\widetilde{s_{ij}}^{\mathrm{comm}_c}$ is the   statistical relationship between individuals $i$ and $j$ in community $c$ in $\widetilde{\mathbf{R}}(\mathrm{FP},\emptyset,\emptyset)$, whose $i$th data record is denoted as $\widetilde{\boldsymbol{r}_i}$, 
and  $\widetilde{\widetilde{s_{ij}}}^{\mathrm{comm}_c}$ is such information  obtained from  $\widetilde{\mathbf{R}}(\mathrm{FP},\mathrm{Dfs_{row}}(\mathcal{S}'),\emptyset)$, whose $i$th data record is represented as $\widetilde{\widetilde{\boldsymbol{r}_i}}$. Also, $\mathrm{value\ change}(\cdot)$ is the function that changes each attribute of $\widetilde{\boldsymbol{r}_i}$, and it will be elaborated later. In (\ref{eq:row_defense}), we let the  cardinality of $\mathcal{E}_c$ to  be smaller than $\left \lceil{n_c\gamma}\right \rceil$ ($\gamma$ is the percentage of fingerprinted records)  to   restrict the number of selected non-fingerprinted records to maintain database utility.

(\ref{eq:row_defense}) is an NP-hard  combinatorial search problem \cite{bertsimas1997introduction}. 
Thus,  we use a greedy algorithm to determine $\mathcal{E}_c$ and a heuristic approach to obtain $\widetilde{\widetilde{\boldsymbol{r}_i}}, i\in\mathcal{E}_c$.   In fact, (\ref{eq:row_defense}) also belongs to the problems of set function maximization, which can be connected  to submodular optimization \cite{wei2015submodularity}, and     greedy algorithms are widely used for   selecting candidate sets. Specifically, Alice   constructs     $\mathcal{E}_c$ by greedily      choosing up to  $\left \lceil{n_c\gamma}\right \rceil$   non-fingerprinted  data records (in $\widetilde{\mathbf{R}}(\mathrm{FP},\emptyset,\emptyset)$) that  have the maximum cumulative absolute difference (i.e., $\sum_{j\in  \mathrm{comm}_c,j\neq i}^{n_c}  \big|{s'_{ij}}^{\mathrm{comm}_c}-  \widetilde{s_{ij}}^{\mathrm{comm}_c}\big|,i\in\mathrm{comm}_c/\mathcal{Q}_c$) with Alice's prior knowledge ($\mathcal{S}'$). 
Next, she changes   the value of each attribute  of the selected data records in  $\mathcal{E}_c$ to the most frequent occurring instance of that attribute to obtain $\widetilde{\widetilde{\boldsymbol{r}_i}}$ (i.e.,  $\widetilde{\widetilde{\boldsymbol{r}_i}} = \mathrm{value\  change}(\widetilde{\boldsymbol{r}_i})$).
We   describe the steps to apply $\mathrm{Dfs_{row}}(\mathcal{S}')$ in Algorithm \ref{algo:row-wise-defense}. 

\begin{algorithm}
\small
\SetKwInOut{Input}{Input}
\SetKwInOut{Output}{Output}
\Input{Vanilla fingerprinted database, $\widetilde{\mathbf{R}}(\mathrm{FP},\emptyset,\emptyset)$, fingerprinting ratio $\gamma$, database owner's prior knowledge on the row-wise correlations  $\mathcal{S}'$ and individuals' affiliation to the $C$ communities.}
\Output{$\widetilde{\mathbf{R}}\Big(\mathrm{FP},\mathrm{Dfs_{row}}(\mathcal{S}'),\emptyset,\Big)$.}

Obtain $\widetilde{\mathcal{S}}$, i.e.,  the   set of  pairwise statistical relationships among individuals in each community,   from the vanilla fingerprinted database $\widetilde{\mathbf{R}}(\mathrm{FP},\emptyset,\emptyset)$;

\ForAll{$\mathrm{comm}_c, c\in[1,C]$}{\ForAll{non-fingerprinted individual $i\in \mathrm{comm}_c/\mathcal{Q}_c$}{

Calculate 
$e_i = \sum_{j\in  \mathrm{comm}_c,j\neq i}^{n_c}  \big|{s'_{ij}}^{\mathrm{comm}_c}-  \widetilde{s_{ij}}^{\mathrm{comm}_c}\big|, i\in\mathrm{comm}_c/\mathcal{Q}_c$;

}

Obtain the largest $\left \lceil{n_c\gamma}\right \rceil$ $e_i$'s, and collect these row index $i$ in set $\mathcal{E}_c$;

\ForAll{row index $i\in\mathcal{E}_c$}{

$\widetilde{\widetilde{\boldsymbol{r}_i}} = \mathrm{value\  change}(\widetilde{\boldsymbol{r}_i})$;
\CommentSty{//change the value of each attribute of $\widetilde{\boldsymbol{r}_i}$ to the most frequently  occurred instance of that attribute in $\mathrm{comm}_c$.\\}
}

}

 Return $\widetilde{\mathbf{R}}\Big(\mathrm{FP},\mathrm{Dfs_{row}}(\mathcal{S}'),\emptyset,\Big)$.

		\caption{$\mathrm{Dfs_{row}}(\mathcal{S}')$: defense against row-wise correlation attack.}
\label{algo:row-wise-defense}
\end{algorithm}

The   solution to (\ref{eq:row_defense}) depends on the database and the distribution of data entries, thus, it is infeasible to derive a generic closed-form expression to  quantify the mitigation performance of $\mathrm{Dfs_{row}}(\mathcal{S}')$. However, in Section \ref{sec:exp_defense_on_census}, we will empirically show that the fraction of the fingerprinted entries inferred by $\mathrm{Atk_{row}}(\mathcal{S})$ will decrease significantly if Alice applies the post-processing step $\mathrm{Dfs_{row}}(\mathcal{S}')$.


\subsection{Integrated Robust Fingerprinting}
Although after applying  $\mathrm{Dfs_{row}}(\mathcal{S}')$, the malicious SP may still identify (and distort) some fingerprinted data records using $\mathrm{Atk_{row}}(\mathcal{J})$, the amount of distortion in the fingerprint will not be enough to compromise the fingerprint bit-string due to the majority voting considered in the vanilla scheme. In Section \ref{sec: int_attack_census}, we validate that Algorithm \ref{algo:row-wise-defense} can successfully mitigate the row-wise correlation attack in a real-world  database. Since $\mathrm{Dfs_{row}}(\mathcal{S}')$ changes less number of entries than $\mathrm{Dfs_{col}}(\mathcal{J}')$, database owner will apply $\mathrm{Dfs_{row}}(\mathcal{S}')$ first after the vanilla fingerprinting.  In Algorithm \ref{algo:robust-fp}, we summarize the main steps of our integrated robust fingerprinting scheme against the identified correlation attacks.

\begin{algorithm}
\small
\SetKwInOut{Input}{Input}
\SetKwInOut{Output}{Output}
\Input{A database $\mathbf{R}$, a vanilla fingerprinting scheme $\mathrm{FP}$, database owner's prior knowledge on the column-wise and row-wise correlation, i.e., $\mathcal{J}'$ and $\mathcal{S}'$, and individuals' affiliation to the $C$ communities.}

\Output{$\widetilde{\mathbf{R}}\Big(\mathrm{FP},\mathrm{Dfs_{row}}(\mathcal{S}'),\mathrm{Dfs_{col}}(\mathcal{J}')\Big)$.}

Apply the vanilla fingerprinting scheme on $\mathbf{R}$ and obtain $\widetilde{\mathbf{R}}\Big(\mathrm{FP},\emptyset,\emptyset\Big)$;

Apply  $\mathrm{Dfs_{row}}(\mathcal{S}')$  on  $\widetilde{\mathbf{R}}\Big(\mathrm{FP},\emptyset,\emptyset\Big)$ using Algorithm \ref{algo:row-wise-defense} and obtain $\widetilde{\mathbf{R}}\Big(\mathrm{FP},\mathrm{Dfs_{row}}(\mathcal{S}'),\emptyset\Big)$;

Apply  $\mathrm{Dfs_{col}}(\mathcal{J}')$ on $\widetilde{\mathbf{R}}\Big(\mathrm{FP},\mathrm{Dfs_{row}}(\mathcal{S}'),\emptyset\Big)$ using Algorithm \ref{algo:ot_defense} and obtain $\widetilde{\mathbf{R}}\Big(\mathrm{FP},\mathrm{Dfs_{row}}(\mathcal{S}'),\mathrm{Dfs_{col}}(\mathcal{J}')\Big)$;

		\caption{Robust fingerprinting against correlation attacks.}
\label{algo:robust-fp}
\end{algorithm}

\section{Evaluation}
\label{sec:eva}
Now, we evaluate the   correlation attacks and the   robust fingerprinting mechanisms, investigate their impact on fingerprint robustness and  database utility, and empirically study  the effect of   knowledge asymmetry between the  database owner and a malicious SP.

\subsection{Experiment Setup}

We consider   a Census database \cite{asuncion2007uci} as the study case. 
As discussed in Section \ref{sec:vanilla}, we choose the state-of-the-art scheme developed in \cite{li2005fingerprinting} as the vanilla mechanism, because it is shown to be robust against common attacks (such as random bit flipping, subset, and superset attacks).
We use 128-bits fingerprint   string  ($L=128$) for    the vanilla   scheme, because when considering $N$ SPs, as long as $L> \ln N$, the vanilla scheme can thwart exhaustive search and various types of attacks \cite{li2005fingerprinting}, and in most cases a 64-bits fingerprint string is shown to provide high robustness. 


In different experiments, to distinguish different instances of the row-wise and column-wise correlations, we also parametrize $\mathcal{J}'$, $\mathcal{S}'$, $\mathcal{J}$, and $\mathcal{S}$  when specifying their resources. For instance,    $\mathcal{J}'(\mathbf{R})$   indicates Alice's prior knowledge on column-wise correlations   are  calculated directly from the original database.



\subsection{Evaluations on Census Database}
\label{sec:census_study}


Census database \cite{asuncion2007uci}  records 14 discrete or categorical attributes of 32561 individuals.  To add fingerprint to this database, Alice first encodes the values of each attribute as integers in a way that the LSB carries the least information. {\color{black}Recall that to achieve high database utility, we let the vanilla scheme only fingerprint the LSBs (in Appendix \ref{sec:utility_LSB} we validate that fingerprinting the other bits reduces database utility).} In particular, for a discrete numerical attribute (e.g., age), the values are first sorted in an ascending order and then divided into   non-overlapping ranges, which are then encoded as ascending  integers starting from 0. For a categorical attribute (e.g., marital-status),  the instances are first mapped to a high dimensional space via the word embedding technique \cite{mikolov2013distributed}. Words having similar meanings appear roughly in the same area of the   space. After mapping, these vectors are clustered into a hierarchical tree structure, where each leaf node represents an instance of that attribute and is encoded by an integer and the adjacent leaf nodes differ in the LSB.    Besides, we use K-means   to group the individuals in the Census database into non-overlapping communities, and according to the Schwarz's Bayesian inference criterion (BIC) \cite{schwarz1978estimating}, the optimal number of communities is $C=10$.

\begin{table*}[htb]
\begin{center}
 \begin{tabular}{c|c | c|c| c | c | c| c| c| c| c  } 
 \hline
 \hline
 Attack on           &   robustness  \&  & using &  \multicolumn{8}{c}{    rounds of $\mathrm{Atk_{col}}\Big(    \mathcal{J}(\mathbf{R}) \Big)$} \\
\cline{4-11}
 $\widetilde{\mathbf{R}}(\mathrm{FP},\emptyset,\emptyset)$ &   utility loss & $\mathrm{Atk_{row}}\Big(   \mathcal{S}(\mathbf{R}) \Big)$  & 1&2&3&4&5&6&7&8 \\
 \hline
using & $\mathrm{num_{cmp}}$ & \multirow{3}{*}{N/A }  &  28  &  43 &   55   & 58   & 63   & 74 &   77  & \cellcolor{Green}82  \\ 
\multirow{2}{*}{$\mathrm{Atk_{col}}\Big(\mathcal{J}(\mathbf{R})\Big)$} & $r$  &    & u     &      u     &      u    & <  0.08\%    &  <  0.73\%   & <  53.2\%   & < 71.8\%   & \cellcolor{Green}$< 91.4\%$   \\
  & $\mathrm{per_{chg}}$ & & \cellcolor{Gray}$4.4\%$ & \cellcolor{Gray}$10.9\%$ & \cellcolor{Gray}$14.2\%$ &  \cellcolor{Gray}$15.9\%$ &  \cellcolor{Gray}$18.4\%$ & \cellcolor{Gray}$23.7\% $ & \cellcolor{Gray}$25.3\%$ & \cellcolor{Gray}$27.1\%$\\
  \hline
using $\mathrm{Atk_{row}}\Big(   \mathcal{S}(\mathbf{R}) \Big)$ & $\mathrm{num_{cmp}}$ &     \cellcolor{Blue}78   &   78  &  79   & 80    &81   & 82  &  83  &  83  &  \cellcolor{Red}83 \\ 
  \multirow{2}{*}{ and  $\mathrm{Atk_{col}}\Big(\mathcal{J}(\mathbf{R})\Big)$} &  $r$    & \cellcolor{Blue}< $82.9\%$  &  < 82.9\%   & <  89.1\%  & < 89.4\%  & < 90.1\%    &<  91.4\%  & < 93.7\%  & < 93.7\%  & \cellcolor{Red}< 93.7\%   \\
 & $\mathrm{per_{chg}}$ & \cellcolor{Blue}2.9\% & 8.9\% & 10.5\%  & 11.3\% & 11.5\% & 11.9\% & 12.6\% & 13.7\% & \cellcolor{Red}14.2\%\\
  \hline
  \hline
\end{tabular}
\end{center}
\caption{Fingerprint robustness and utility loss of different correlation attacks on the vanilla fingerprinted Census database $\widetilde{\mathbf{R}}(\mathrm{FP},\emptyset,\emptyset)$. The fingerprint robustness  metrics are the  number of compromised fingerprint bits, i.e.,  $\mathrm{num_{cmp}}$ and i.e., accusable ranking $r$. The   utility loss of the malicious SP is the  fraction of modified entries as a result of the attack, i.e., $\mathrm{per_{chg}} = 1-Acc(\overline{\mathbf{R}})$. `u' stands for uniquely accusable. `< $r$' means top $r$ accusable.} 
\label{table:correlation_attack_compare}
\end{table*}

\subsubsection{Impact of Correlation Attacks on Census Database}
\label{sec: int_attack_census}

We first study  the impact  of $\mathrm{Atk_{row}}(\mathcal{S})$ and $\mathrm{Atk_{col}}(\mathcal{J})$, and then present the impact of the integration of them.  In this experiment, we assume that the malicious SP  has the ground truth knowledge about the row- and column-wise correlations, i.e., it has access to $\mathcal{S}$ and $\mathcal{J}$ that  are directly computed from $\mathbf{R}$. 
As a result, we represent its prior knowledge as  $\mathcal{S}(\mathbf{R})$ and $\mathcal{J}(\mathbf{R})$. By launching  the  row-wise,  column-wise, and integrated correlation attack, the  malicious SP   generates  pirated database      $\overline{\mathbf{R}}\Big(\mathrm{FP},\mathrm{Atk_{row}}(\mathcal{S}(\mathbf{R})),\emptyset)\Big)$,   $\overline{\mathbf{R}}\Big(\mathrm{FP},\emptyset,\mathrm{Atk_{col}}(\mathcal{J}(\mathbf{R}))\Big)$, and $\overline{\mathbf{R}}\Big(\mathrm{FP},\mathrm{Atk_{row}}(\mathcal{S}(\mathbf{R})),\mathrm{Atk_{col}}(\mathcal{J}(\mathbf{R}))\Big)$, respectively.

\noindent\textbf{Impact of $\mathrm{Atk_{col}}(\mathcal{J}(\mathbf{R}))$.} 
First, we validate that $\mathrm{Atk_{col}}(\mathcal{J}(\mathbf{R}))$ is more powerful than the   random bit flipping attack $\mathrm{Atk_{rnd}}$ discussed in Section \ref{sec:system-threat}. 
We set the threshold $\tau_{\mathrm{col}}^{\mathrm{Atk}} = 0.0001$ when comparing  $|J_{p,q}(a,b)-\widetilde{J_{p,q}}(a,b)|$.\footnote{In all experiments, we choose a small value for $\tau_{\mathrm{col}}^{\mathrm{Atk}}$, $\tau_{\mathrm{col}}^{\mathrm{Dfs}}$, and $\tau_{\mathrm{col}}$, because a database usually contains     thousands of data records and the addition of fingerprint changes a small fraction of entries, which does not cause large changes in the joint distributions. On the contrary, we choose a large value for $\tau_{\mathrm{row}}^{\mathrm{Atk}}$  and $\tau_{\mathrm{row}}$, because the statistical relationship is defined as an exponentially decay function, which ranges from 0 to 1, and the added fingerprint results in a larger change for this  statistical relationship.} As a result, it takes 8 iterations (attack rounds) for $\mathrm{Atk_{col}}(\mathcal{J}(\mathbf{R}))$ to converge (i.e., stop including new suspicious fingerprinted positions in $\mathcal{P}$). In Table \ref{table:correlation_attack_compare}, we record the fingerprint robustness  (i.e.,  $\mathrm{num_{cmp}}$ and   $r$) and utility loss of the malicious SP (fraction of modified entries as a result of the attack, i.e.,  $\mathrm{per_{chg}} = 1-Acc(\overline{\mathbf{R}})$) when launching increasing rounds of $\mathrm{Atk_{col}}(\mathcal{J}(\mathbf{R}))$ 
 on the vanilla fingerprinted Census database. We observe that with more attack rounds, more fingerprint bits are compromised, and the accusable ranking of the malicious SP also decreases, which suggests that Alice may accuse innocent SP with increasing probability.   In Table \ref{table:atk_rnd}, we present the performance of $\mathrm{Atk_{rnd}}$ on the vanilla  fingerprinted   database. Specifically, by setting the fraction of entries changed ($\mathrm{per_{chg}}$) due to $\mathrm{Atk_{rnd}}$  equal to that of $\mathrm{Atk_{col}}(\mathcal{J}(\mathbf{R}))$ with increasing rounds (i.e., the cells highlighted in gray in Table \ref{table:correlation_attack_compare}), we calculated $\mathrm{num_{cmp}}$ and $r$ achieved by $\mathrm{Atk_{rnd}}$. 

\begin{table}[htb]
\begin{center}
 \begin{tabular}{c | c|c| c | c | c| c } 
 \hline
 \hline
 $\mathrm{per_{chg}}$&    $\leq14.2\%$ &  $15.9\%$ &  $18.4\%$ & $23.7\% $ & $25.3\%$ &  $27.1\%$ \\
 \hline
 $\mathrm{num_{cmp}}$ &  0    &1   & 1  &  2  &  3  &   4 \\ 
$r$ &    u   & u    & u  &u   & u  & u    \\
  \hline
  \hline
\end{tabular}
\caption{Performance and cost of  $\mathrm{Atk_{rnd}}$ on the vanilla fingerprinted Census database. $\mathrm{per_{chg}}$ values are set to be equal to that of  $\mathrm{Atk_{col}}(\mathcal{J}(\mathbf{R}))$  (cells highlighted in gray in Table \ref{table:correlation_attack_compare}). $r=\text{u}$ means uniquely accusable.} 
\label{table:atk_rnd}
\end{center}
\end{table}

Combining  Tables \ref{table:correlation_attack_compare}  and \ref{table:atk_rnd}, we observe that if $\mathrm{per_{chg}}$ is below $14.2\%$, $\mathrm{Atk_{rnd}}$ cannot compromise any fingerprint bits, whereas $\mathrm{Atk_{col}}(\mathcal{J}(\mathbf{R}))$ compromises 28 fingerprint bits (out of 128).  Even when $\mathrm{per_{chg}} = 27.1\%$, $\mathrm{Atk_{rnd}}$ can only distort 4 fingerprint bits. As a result, if the malicious SP launches $\mathrm{Atk_{rnd}}$, it will be uniquely accusable for pirating the database. 
Whereas, when $\mathrm{per_{chg}} = 27.1\%$  $\mathrm{Atk_{col}}(\mathcal{J}(\mathbf{R}))$ distorts 82 bits, which makes the malicious SP only rank top $91.4\%$ accusable and will cause Alice accuse innocent SP with very high probability (the cells highlighted in green in Table \ref{table:correlation_attack_compare}).  
In fact, for $\mathrm{Atk_{rnd}}$ to compromise enough fingerprint bits so as to cause Alice to accuse innocent SPs, it needs to flip more than $83\%$ of the entries in the fingerprinted Census database. Clearly, the vanilla fingerprint scheme is robust against $\mathrm{Atk_{rnd}}$, however, its robustness significantly degrades against $\mathrm{Atk_{col}}(\mathcal{J}(\mathbf{R}))$.

\begin{figure}
  \begin{center}
  \begin{tabular}{ccc}
     \includegraphics[width= .31\columnwidth,height=2.5cm]{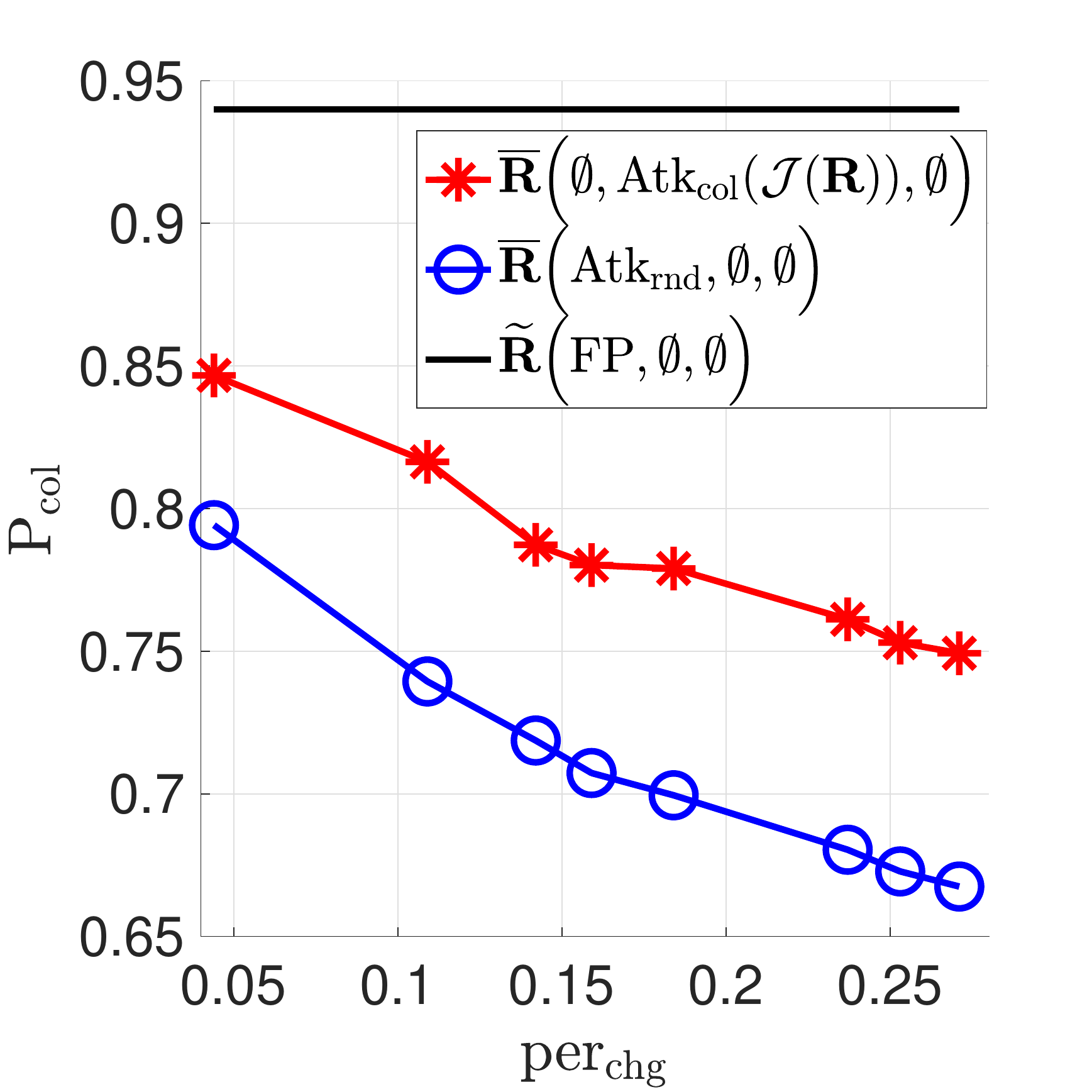}
		&
		     \includegraphics[width= .31\columnwidth,height=2.5cm]{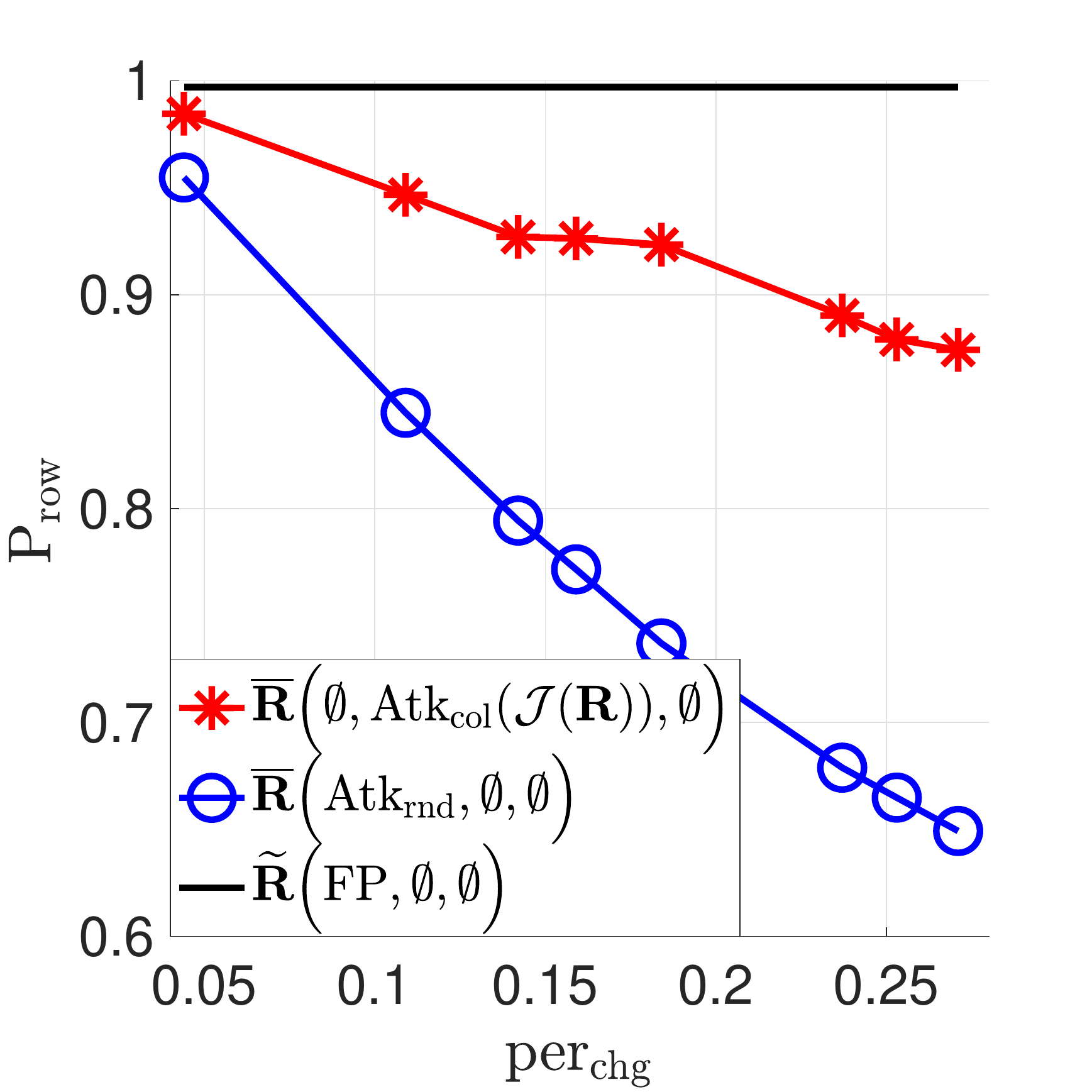}
		&
     \includegraphics[width= .31\columnwidth,height=2.5cm]{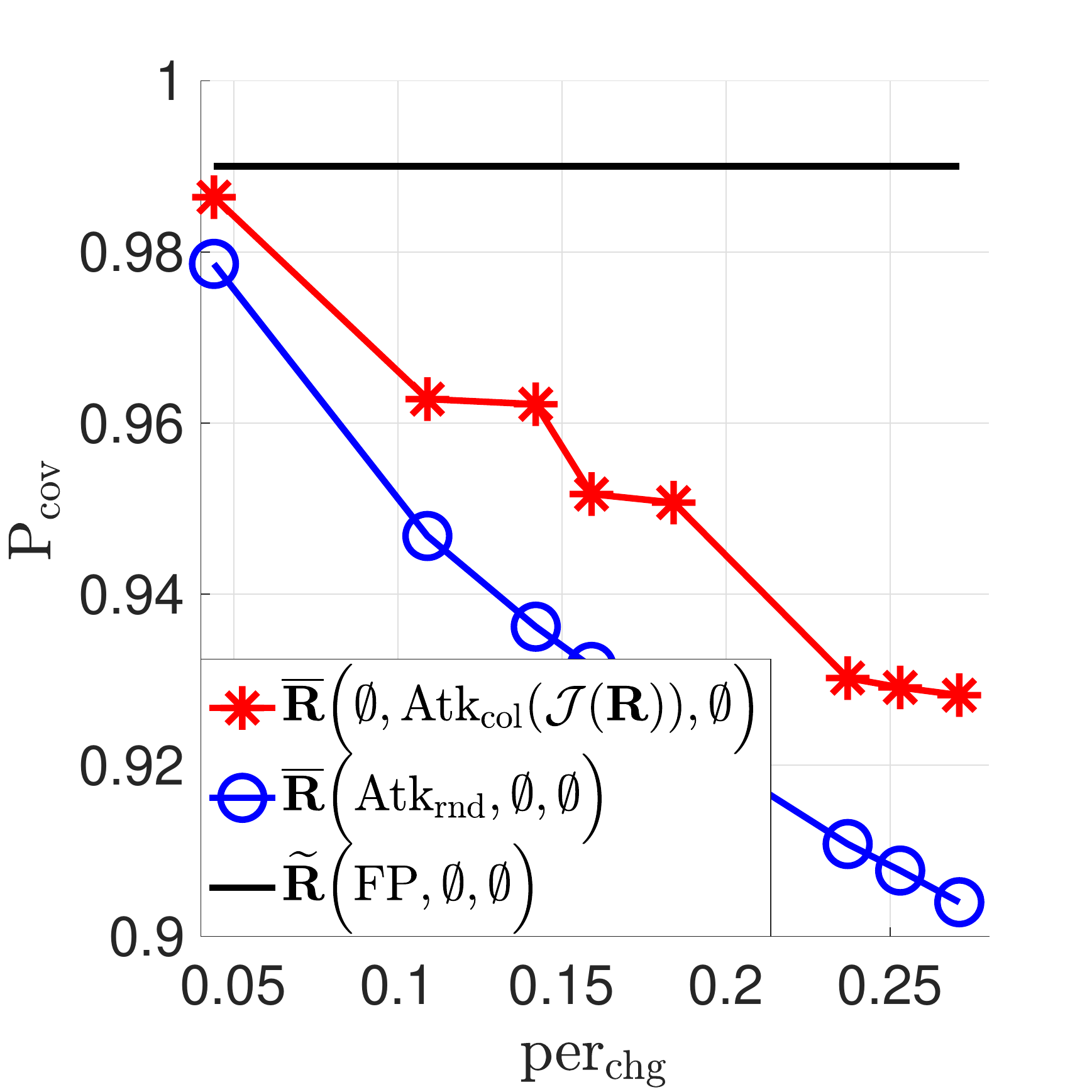}
     \\
    {\small  (a) $\mathrm{P_{col}}$ }  &
    {\small  (b) $\mathrm{P_{row}}$ } &
    {\small  (c) $\mathrm{P_{cov}}$ }\\
        \end{tabular}
      \end{center}
  \caption{\label{fig:atkcol_vs_atkrnd}  Comparison of (a) $\mathrm{P_{col}}$, (b) $\mathrm{P_{row}}$, and (c) $\mathrm{P_{cov}}$ achieved by $\overline{\mathbf{R}}(\emptyset,\emptyset,\mathrm{Atk_{col}}(\mathcal{J}(\mathbf{R})))$ and  $\overline{\mathbf{R}}(\mathrm{Atk_{rnd}},\emptyset,\emptyset)$ when $\mathrm{per_{chg}}$ are set as the values highlighted in gray in Table \ref{table:correlation_attack_compare}.}
\end{figure}

In Figure \ref{fig:atkcol_vs_atkrnd}, by setting  $\mathrm{per_{chg}}$ to the values highlighted in gray in Table \ref{table:correlation_attack_compare}, 
we compute and compare  the    {\color{black}utility of the pirated database (i.e, $\mathrm{P_{col}}(\overline{\mathbf{R}})$ for $\tau_{\mathrm{col}} = 0.0001$, $\mathrm{P_{row}}(\overline{\mathbf{R}})$ for $\tau_{\mathrm{row}} = 10$,  and $\mathrm{P_{cov}}(\overline{\mathbf{R}})$)} obtained from the  vanilla fingerprinted database after $\mathrm{Atk_{rnd}}$ and $\mathrm{Atk_{col}}(\mathcal{J}(\mathbf{R}))$, i.e., $\overline{\mathbf{R}}(\mathrm{Atk_{rnd}},\emptyset,\emptyset)$ and $\overline{\mathbf{R}}(\emptyset,\emptyset,\mathrm{Atk_{col}}(\mathcal{J}(\mathbf{R})))$. We also plot the utility of the vanilla fingerprinted database 
using black lines as the benchmark ($\mathrm{P_{col}}(\widetilde{\mathbf{R}}(\mathrm{FP},\emptyset,\emptyset))=0.95$, $\mathrm{P_{row}}(\widetilde{\mathbf{R}}(\mathrm{FP},\emptyset,\emptyset))=1.00$, and $\mathrm{P_{cov}}(\widetilde{\mathbf{R}}(\mathrm{FP},\emptyset,\emptyset))=0.99$). Clearly, $\overline{\mathbf{R}}(\emptyset,\emptyset,\mathrm{Atk_{col}}(\mathcal{J}(\mathbf{R})))$ always achieves higher utility values than $\overline{\mathbf{R}}(\mathrm{Atk_{rnd}},\emptyset,\emptyset)$, and it has similar utility values compared to $\widetilde{\mathbf{R}}(\mathrm{FP},\emptyset,\emptyset)$ when $\mathrm{per_{chg}}$ is small, e.g., if $\mathrm{per_{chg}}\leq 20\%$,  $\mathrm{P_{col}}\big(\overline{\mathbf{R}}(\emptyset,\emptyset, \mathrm{Atk_{col}}(\mathcal{J}(\mathbf{R})))\big)\geq 0.92$. Combining Table \ref{table:correlation_attack_compare}, \ref{table:atk_rnd}, and Figure \ref{fig:atkcol_vs_atkrnd}, we conclude that $\mathrm{Atk_{col}}(\mathcal{J}(\mathbf{R}))$ is not only more powerful (in terms of distorting the fingerprint bit-string), but it also preserves more database utility compared to $\mathrm{Atk_{rnd}}$. 
{\color{black}In addition to the generic utility metrics defined in Section \ref{sec:utility_metrics}, we also calculate and compare the utility some specific statistical computations on the pirated database. For example, under the same attack performance (i.e., compromising exactly 63 fingerprinting bits) $\mathrm{Atk_{col}}(\mathcal{J})$ only causes $0.3\%$ change in the frequency of individuals   having bachelor degree or higher and $0.01$ change for the standard deviation of individuals' age, whereas, the same values for $\mathrm{Atk_{rnd}}$ are $1.4\%$ and $0.12$, respectively. }

\noindent\textbf{Impact of $\mathrm{Atk_{row}}(\mathcal{S}(\mathbf{R}))$.} 
{\color{black}By setting the threshold $\tau_{\mathrm{row}}^{\mathrm{Atk}} = 0.1$ when comparing $\textstyle\sum_{j\neq i}^{n_c}|s_{ij}^{\mathrm{comm}_c}-\widetilde{s_{ij}}^{\mathrm{comm}_c}|$ in Section \ref{sec:row-wise-atack}} 
we show the impact of $\mathrm{Atk_{row}}(\mathcal{S}(\mathbf{R}))$ in the blue cells of  Table \ref{table:correlation_attack_compare}. After launching  row-wise correlation attack on $\widetilde{\mathbf{R}}(\mathrm{FP},\emptyset,\emptyset)$, 78 fingerprint bits are distorted at the cost of only $2.9\%$ utility loss. It makes the malicious SP only rank top $82.9\%$ accusable, and may cause Alice accuse innocent SP with high probability.  
In particular, we   have  $\mathrm{P_{col}}\big(\overline{\mathbf{R}}(\emptyset,\mathrm{Atk_{col}}(\mathcal{S}(\mathbf{R})),\emptyset)\big) = 0.90$, $\mathrm{P_{row}}\big(\overline{\mathbf{R}}(\emptyset, \mathrm{Atk_{col}}(\mathcal{S}(\mathbf{R})),\emptyset)\big)= 0.95$, and $\mathrm{P_{cov}}\big(\overline{\mathbf{R}}(\emptyset,\mathrm{Atk_{col}}(\mathcal{S}(\mathbf{R})),\emptyset)\big)= 0.97$, which are all closer to that of  $\widetilde{\mathbf{R}}(\mathrm{FP},\emptyset,\emptyset)$. This suggests again that the identified correlation attacks are powerful than the  conventional attacks and they can maintain the utility of database.

\noindent\textbf{Impact of integrated  correlation attack.} 
By launching  $\mathrm{Atk_{row}}(\mathcal{S}(\mathbf{R}))$ on $\widetilde{\mathbf{R}}(\mathrm{FP},\emptyset,\emptyset)$ followed by 8 rounds of $\mathrm{Atk_{col}}(\mathcal{J}(\mathbf{R}))$, the integrated correlation attack can   distort more fingerprint bits, i.e., 83 bits, which makes the malicious SP's accusable ranking drops to top $93.7\%$ (the cells highlighted in red in Table \ref{table:correlation_attack_compare}). This suggests again that the vanilla fingerprint scheme is not capable of identifying the guilty SP that is liable for pirating the database if the malicious SP utilizes data correlations  to distort the fingerprint.

Note that, although $\mathrm{Atk_{col}}(\mathcal{J}(\mathbf{R}))$ has similar attack performance compared to the integrated attack, its utility loss is higher, i.e., $27.1\%$ entries are modified by the attacker. Besides, at the early stages of $\mathrm{Atk_{col}}(\mathcal{J}(\mathbf{R}))$, the malicious SP cannot distort more than half of the fingerprint bits (e.g., at the end of the $5$th round, only 63 bits are compromised by modifying $15.9\%$ of the entries), which is inadequate to cause Alice accuse innocent SPs and also makes the malicious SP uniquely accusable.  Since $\mathrm{Atk_{row}}(\mathcal{S}(\mathbf{R}))$ can distort sufficient  fingerprint bits and cause Alice to accuse innocent SPs with high probability at a much lower utility loss {\color{black}(measured using both generic utility metrics and specific statistical utilities, like the change in frequencies of data records and standard deviations)}, we conclude that it is more powerful than $\mathrm{Atk_{col}}(\mathcal{J}(\mathbf{R}))$. {\color{black}This suggests that in real-world integrated correlation attacks, the malicious SP can   conduct $\mathrm{Atk_{row}}(\mathcal{S}(\mathbf{R}))$  followed by a few rounds of $\mathrm{Atk_{col}}(\mathcal{J}(\mathbf{R}))$ to simultaneously distort a large number of fingerprint bits and preserve data utility when generating the pirated database.}

\subsubsection{Performance of Mitigation Techniques on Census Database} \label{sec:exp_defense_on_census}

We have shown that correlation attacks can distort the fingerprint bit-string and may make the database owner accuse innocent SPs by resulting in low degradation in terms of database utility. In this section, we first evaluate the proposed mitigation techniques against correlation attacks  separately, and then  consider the integrated  mitigation technique  against the integrated correlation attack, i.e., the row-wise correlation attack followed by the column-wise correlation attack. 
In this experiment, we also assume that Alice  has  access to $\mathcal{S}'$ and $\mathcal{J}'$that  are directly computed from $\mathbf{R}$. Thus, we represent her prior knowledge as  $\mathcal{S}'(\mathbf{R})$ and $\mathcal{J}'(\mathbf{R})$.  As a result, we have $\mathcal{S}' = \mathcal{S}$ and $\mathcal{J}' = \mathcal{J}$. 



\noindent\textbf{Performance of $\mathrm{Dfs_{col}}(\mathcal{J}'(\mathbf{R}))$.} As discussed in Section \ref{sec:col_defense}, the mitigation strategy is determined by the marginal probability mass transportation plan, which is heterogeneous for higher $\lambda_p$ (a tuning parameter controlling the entropy of the transportation plan) and homogeneous for lower $\lambda_p$. To evaluate the utility loss due to $\mathrm{Dfs_{col}}(\mathcal{J}'(\mathbf{R}))$, we calculate the utility of  $\widetilde{\mathbf{R}}\big(\mathrm{FP},\emptyset,\mathrm{Dfs_{col}}(\mathcal{J}'(\mathbf{R}))\big)$ by setting $\lambda_p   \in \{100,\cdots,1000\},\forall p\in\mathcal{F}$, and show 
the results in Figure \ref{fig:utility_lambda}. We see that all utilities monotonically increase as the mass transportation plans transform from homogeneous to heterogeneous (i.e., as $\lambda_p$ increases). This is because, as the transportation plans become more heterogeneous, the mitigation technique can tolerate more discrepancy between two marginal distributions (Section \ref{sec:col_defense}), and hence fewer number of entries are modified by $\mathrm{Dfs_{col}}(\mathcal{J}'(\mathbf{R}))$.

\begin{figure}[htb]
  \begin{center}
     \includegraphics[width= 0.6\columnwidth]{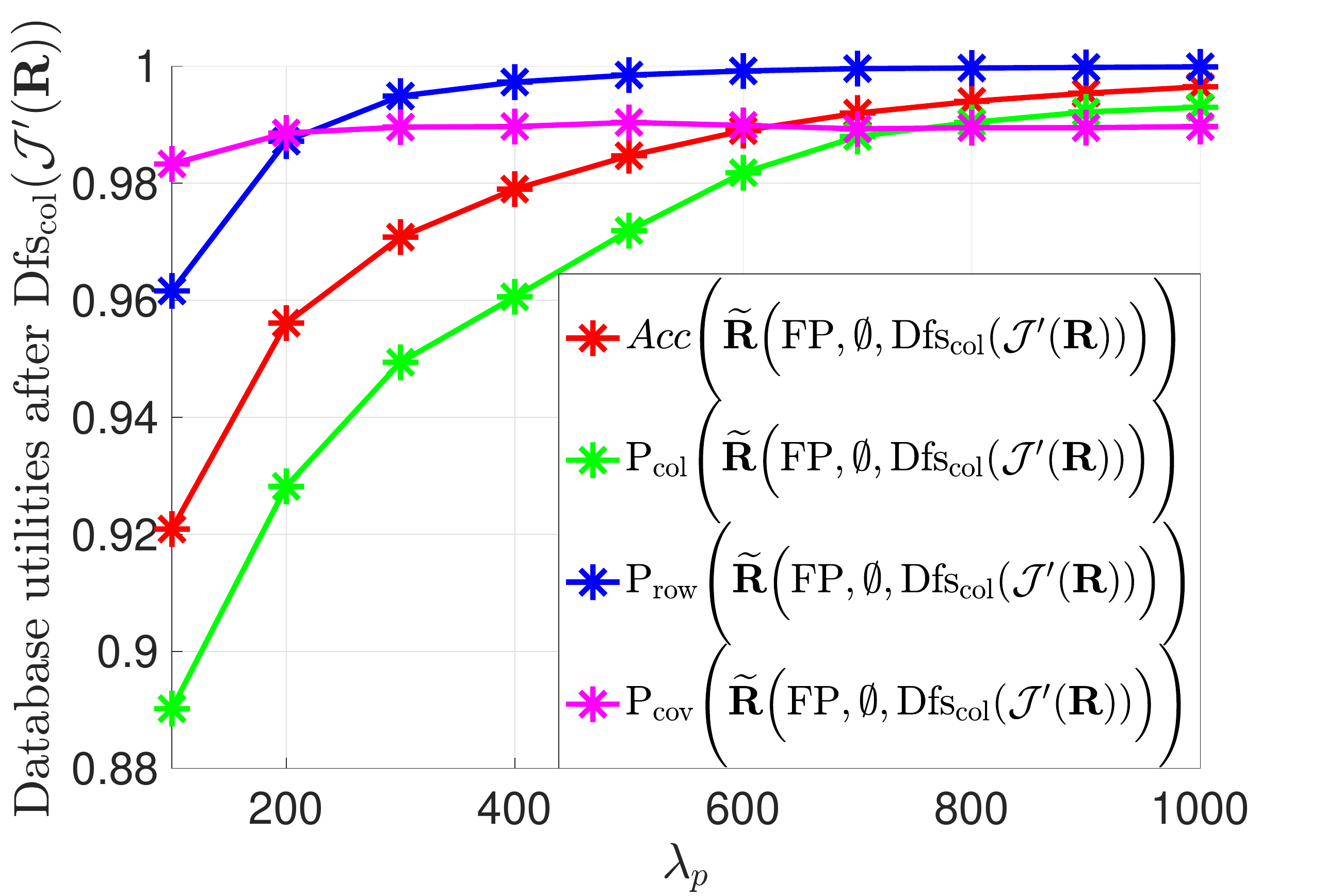}
      \end{center}
  \caption{\label{fig:utility_lambda} Utilities of $\widetilde{\mathbf{R}}(\mathrm{FP},\emptyset,\mathrm{Dfs_{col}}(\mathcal{J}'(\mathbf{R})))$ under varying $\lambda_p$.}
\end{figure}

Next, we fix $\lambda_p  = 500, \forall p\in\mathcal{F}$,   evaluate the performance (in terms of both fingerprint  robustness and database utility) of launching $\mathrm{Atk_{col}}(\mathcal{J}(\mathbf{R}))$  on $\widetilde{\mathbf{R}}\big(\mathrm{FP},\emptyset,\mathrm{Dfs_{col}}(\mathcal{J}'(\mathbf{R}))\big)$ with increasing attack rounds. In Figure \ref{fig:atkcol_after_dfscol}(a), we observe that at then end of 8 rounds of $\mathrm{Atk_{col}}(\mathcal{J}(\mathbf{R}))$, the malicious SP can only compromise 24 (out of 128) fingerprint bits, which is not enough to cause Alice accuse innocent SPs and will make itself uniquely accusable.   In contrast, as shown in Table \ref{table:correlation_attack_compare}, when launching $\mathrm{Atk_{col}}(\mathcal{J}(\mathbf{R}))$ on the vanilla fingerprinted database $\widetilde{\mathbf{R}}(\mathrm{FP},\emptyset,\emptyset)$, the malicious SP can compromise 82 bits and make itself only rank top $91.4\%$ accusable.  This suggests that proposed $\mathrm{Dfs_{col}}(\mathcal{J}'(\mathbf{R}))$ significantly mitigates the column-wise correlation attack.

\begin{figure}[htb]
  \begin{center}
  \begin{tabular}{cc}
     \includegraphics[width= .4\columnwidth]{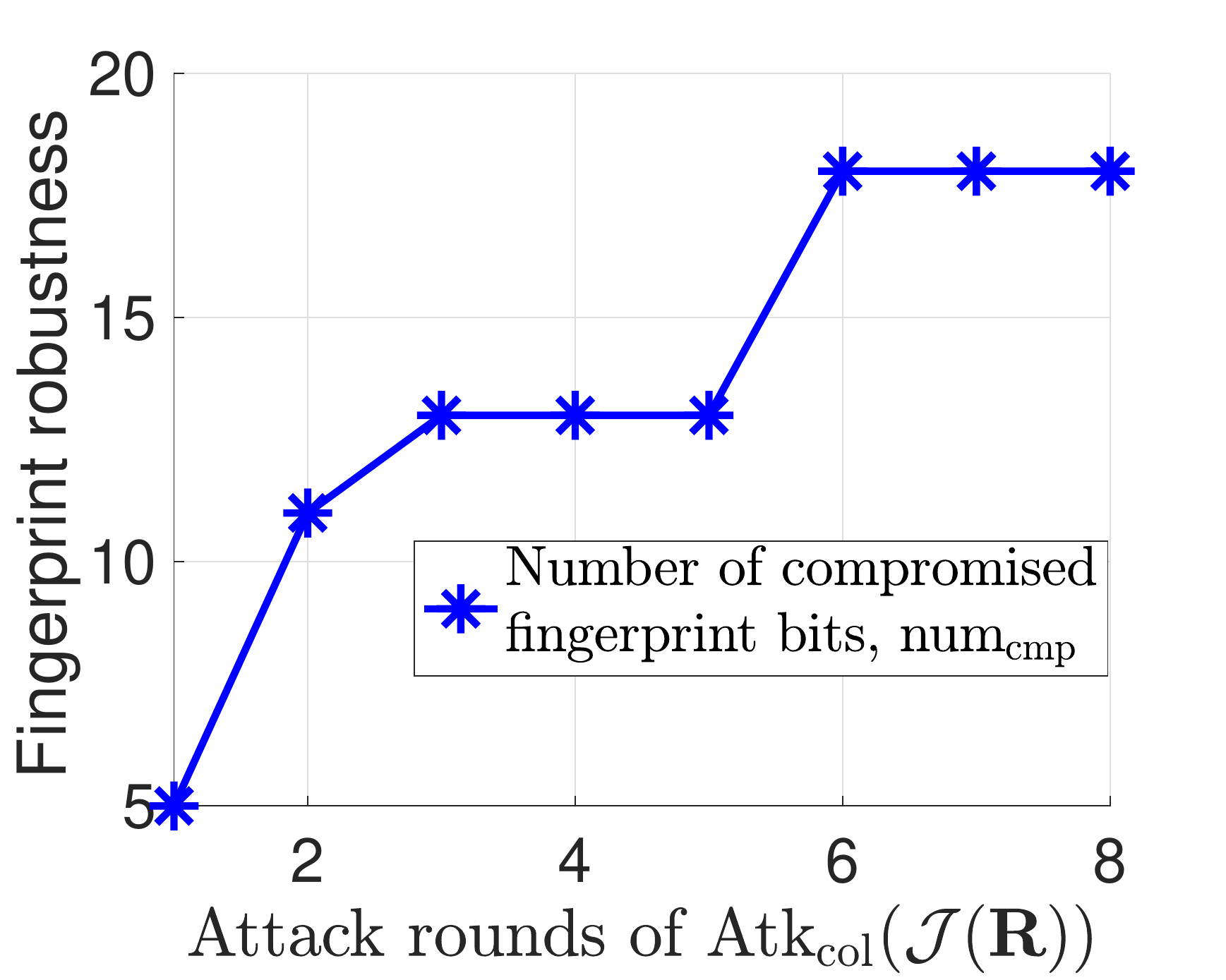}
		&
		     \includegraphics[width= .4\columnwidth]{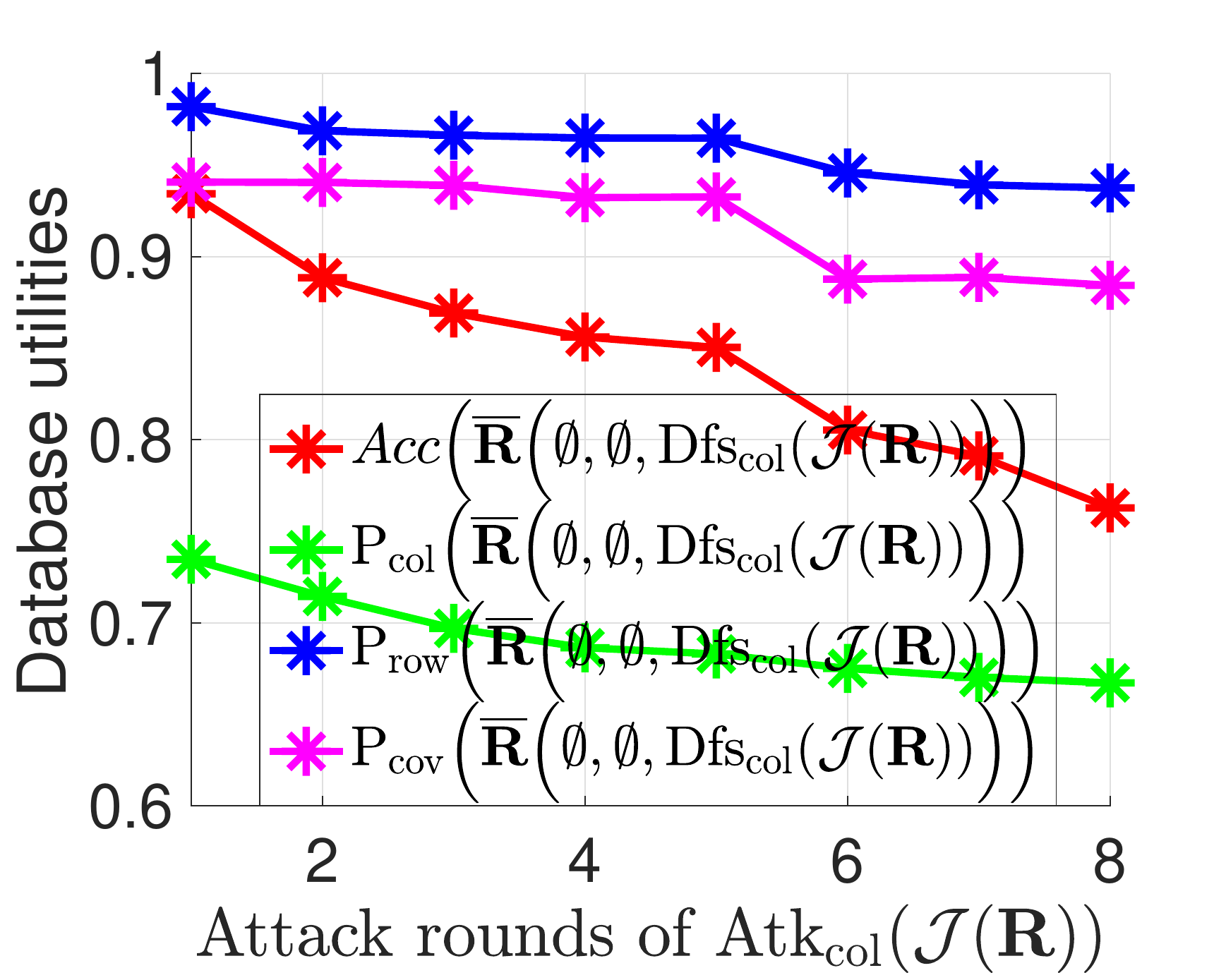}
 \\
    {\small  (a) Fingerprinting robustness}  &
    {\small  (b) Database utilities} \\
    {for increasing attack rounds.} &
     {for increasing attack rounds.} \\
        \end{tabular}
      \end{center}
  \caption{\label{fig:atkcol_after_dfscol} Fingerprint robustness and database utilities when launching $\mathrm{Atk_{col}}(\mathcal{J}(\mathbf{R})))$ on $\widetilde{\mathbf{R}}(\mathrm{FP},\emptyset,\mathrm{Dfs_{col}}(\mathcal{J}'(\mathbf{R})))$.}
\end{figure}

Furthermore, in Figure \ref{fig:atkcol_after_dfscol}(b) we observe that $\mathrm{Atk_{col}}(\mathcal{J}(\mathbf{R}))$ also degrades the utilities of the vanilla fingerprinted database post-processed by $\mathrm{Dfs_{col}}(\mathcal{J}'(\mathbf{R}))$. In particular, the accuracy drops to 0.76 and the preservation of column-wise correlation drops to 0.67 at the end of 8 rounds of $\mathrm{Atk_{col}}(\mathcal{J}(\mathbf{R}))$. Combining  Figures \ref{fig:utility_lambda} and \ref{fig:atkcol_after_dfscol}, we conclude that, as a post-processing step, the proposed column-wise correlation mitigation technique provides robust fingerprint against column-wise correlation attack and preserves   database utility.

\noindent\textbf{Performance of $\mathrm{Dfs_{row}}(\mathcal{S}'(\mathbf{R}))$.} In Table \ref{table:performance-row-defense}, we evaluate the performance of the robust fingerprinted database against row-wise attack, i.e., $\widetilde{\mathbf{R}}(\mathrm{FP},\mathrm{Dfs_{row}}(\mathcal{S}'(\mathbf{R}),\emptyset)$, along with the pirated database obtained by launching $\mathrm{Atk_{row}}(\mathcal{S}(\mathbf{R})$ on it. Clearly, $\mathrm{Dfs_{row}}(\mathcal{S}'(\mathbf{R}))$ successfully defends against $\mathrm{Atk_{row}}(\mathcal{S}(\mathbf{R}))$, since the pirated database only distorts 13 fingerprint bits and makes the malicious SP uniquely accusable. Combining this result with Table \ref{table:correlation_attack_compare} (cells in blue), we conclude that $\mathrm{Dfs_{row}}(\mathcal{S}'(\mathbf{R}))$ not only mitigates the row-wise correlation attack but it also preserves the database utility. 

\begin{table}[htb]
\begin{center}
 \begin{tabular}{c | c|c| c | c | c| c } 
 \hline
 \hline
  &    $Acc$ &  $\mathrm{P_{col}}$  & $\mathrm{P_{row}}$ & $\mathrm{P_{cov}}$ & $\mathrm{num_{cmp}}$ &  $r$ \\
 \hline
$\widetilde{\mathbf{R}}(\mathrm{FP},\mathrm{Dfs_{row}}(\mathcal{S}'),\emptyset)$  &  0.97    & 0.94   & 0.99  &  0.99 & N/A  &   N/A \\ 
$\overline{\mathbf{R}}(\emptyset,\mathrm{Atk_{row}}(\mathcal{S}),\emptyset)$ &    0.93 & 0.92  & 0.94& 0.98 & 13& u  \\
  \hline
  \hline
\end{tabular}
\caption{Impact of  $\mathrm{Dfs_{row}}(\mathcal{S}'(\mathbf{R}))$ before and after $\mathrm{Atk_{row}}(\mathcal{S}(\mathbf{R}))$. $r=\text{u}$ means uniquely accusable.} 
\label{table:performance-row-defense}
\end{center}
\end{table}

\noindent\textbf{Performance of integrated mitigation.} Here, we investigate the performance of the integrated mitigation  against the integrated correlation attacks. By setting $\lambda_p=500, \forall p \in \mathcal{F}$, we evaluate the utility of $\widetilde{\mathbf{R}}\big(\mathrm{FP},\mathrm{Dfs_{row}}(\mathcal{S}'(\mathbf{R})),\mathrm{Dfs_{col}}(\mathcal{J}'(\mathbf{R}))\big)$ before and after it is subject to the integrated attack, i.e.,  $\mathrm{Atk_{row}}(\mathcal{S}(\mathbf{R}))$ followed by $\mathrm{Atk_{col}}(\mathcal{J}(\mathbf{R}))$. We show the results in Table \ref{table:int-dfs-int-atk}. Clearly,  after   integrated mitigation, the fingerprinted database still maintains high utilities. Even if the malicious SP launches  integrated correlation attack, it can only compromise 4 fingerprint bits and  makes itself uniquely accusable. It suggests that the proposed mitigation techniques provide high robustness against integrated correlated attacks.

\begin{table}[htb]
\begin{center}
 \begin{tabular}{c | c|c| c | c | c| c } 
 \hline
 \hline
  &    $Acc$ &  $\mathrm{P_{col}}$  & $\mathrm{P_{row}}$ & $\mathrm{P_{cov}}$ & $\mathrm{num_{cmp}}$ &  $r$ \\
 \hline
after int. mitigation&  0.94    & 0.91   & 0.96  &  0.97 & N/A  &   N/A \\ 
after int. attack&    0.77 & 0.82  & 0.86 & 0.94 & 4& u \\
  \hline
  \hline
\end{tabular}
\caption{Impact of  integrated mitigation    before and after   integrated correlation attack. $r=\text{u}$ means uniquely accusable.} 
\label{table:int-dfs-int-atk}
\end{center}
\end{table}

\section{Conclusion}
\label{sec:conclusion}

In this paper, we have proposed robust fingerprinting for relational databases. First, we have validated  the vulnerability of existing database fingerprinting schemes by identifying different correlation attacks: column-wise correlation attack (which utilizes the joint distributions among attributes), row-wise correlation attack (which utilizes the statistical relationships among the rows), and integration of them. Next, to defend against the identified attacks, we have  developed mitigation techniques that can work as post-processing steps for any off-the-shelf database fingerprinting schemes. Specifically, the column-wise mitigation technique modifies limited entries in the fingerprinted database by solving a set of optimal mass transportation problems concerning pairs of marginal distributions. On the other hand, the row-wise  mitigation technique modifies a small fraction of the fingerprinted database entries by solving a combinatorial search problem. We have also empirically investigated the impact of the identified correlation attacks and the performance of proposed mitigation techniques on an real-world relational database. Experimental results show high success rates for the correlation attacks and high robustness for the proposed mitigation techniques, which alleviate the  attacks   having access to correlation models directly calculated from the data.


\begin{acks}
Research reported in this publication was supported by the National Library Of Medicine of the National Institutes of Health under Award Number R01LM013429.
\end{acks}


\bibliographystyle{ACM-Reference-Format}
\bibliography{fingerprinting}


\appendix

\section{Tradeoff Between Fingerprint Robustness and Database Utility}\label{sec:utility_LSB}
As discussed in Section \ref{sec:vanilla}, to preserve database utility, the added fingerprint only changes the LSB of database entries. In this experiment, we show that if the fingerprint bits are embedded into other bits of entries, some utility metrics will decrease. Specifically, by fixing the fingerprinting ratio to $1/30$, we evaluate the utility (e.g., preservation of correlations and statistics metrics) of the fingerprinted Census database obtained by using the vanilla fingerprinting scheme and changing one of the least $k$ ($k\geq 2$) significant bits (i.e., L$k$SB) of database entries (to add the fingerprint). We show the results in Table \ref{table:LkSB}.


\vspace{-3mm}

\begin{table}[H]
\begin{center}
 \begin{tabular}{c | c | c  | c |c}   
 \hline
 \hline
Utilities & LSB & L2SB & L3SB & L4SB \\
  \hline
  $Acc\Big(\widetilde{\mathbf{R}}(\mathrm{FP},\emptyset,\emptyset)\Big)$ &  0.98 &  0.98& 0.98 & 0.98\\
 $\mathrm{P_{col}}\Big(\widetilde{\mathbf{R}}(\mathrm{FP},\emptyset,\emptyset)\Big)$ & 0.95 & 0.90  & 0.88 & 0.86\\
 $\mathrm{P_{row}}\Big(\widetilde{\mathbf{R}}(\mathrm{FP},\emptyset,\emptyset)\Big)$ & 1.00 & 0.98  & 0.98 & 0.98\\
  $\mathrm{P_{cov}}\Big(\widetilde{\mathbf{R}}(\mathrm{FP},\emptyset,\emptyset)\Big)$ & 0.99 & 0.96 & 0.95 &0.94\\
     \hline
 \hline
\end{tabular}
\caption{Different database utility values obtained when the insertion of  fingerprint changes one of the least $k$ significant bits of database entries.} 
\label{table:LkSB}
\end{center}
\end{table}

\vspace{-4mm}
We observe that although all fingerprinted databases achieve the same accuracy  when the fingerprinting ratio is set to be $1/30$, other utilities decrease if the added fingerprint changes L$k$SB ($k\geq 2$) of data entries. Especially, the preservation of column-wise correlation degrades the most as $k$ increases. The reason is that some pairs of attributes are highly correlated and changing one of the L$k$SB may create statistical unlikely pairs. 
For example,   Masters education degree corresponds to education of 14 years, if the L4SB of 14 (``1110'') is flipped, we end up with an individual who has a master degree with only 6 (``0110'') years of education, which compromise the correlation between ``education'' and ``education-num''.


\end{document}